\documentclass[envcountsame,envcountsect,vecphys]{svmult}
\usepackage{amsmath}
\usepackage{amssymb}
\usepackage{bm}
\usepackage{mathptmx}       
\usepackage{helvet}         
\usepackage{courier}        
\usepackage{type1cm}        
\usepackage{graphicx}
\usepackage[bottom]{footmisc}
\usepackage{cite}
\spnewtheorem{result}[theorem]{Result}{\bfseries}{\itshape}
\spnewtheorem*{definition*}{Definition}{\bfseries}{\rmfamily}
\newcommand{\R}{{\mathbb R}}
\renewcommand{\H}{{\mathbb H}}
\newcommand{\Z}{{\mathbb Z}}
\newcommand{\inner}[2]{{\langle{#1},{#2}\rangle}}
\newcommand{\name}[1]{{\upshape #1}}
\newcommand{\net}[1]{{\bfseries #1}}
\newcommand{\group}[1]{$#1$}
\newcommand{\vv}[1]{{\bm #1}}
\newcommand{\M}[1]{{\bm #1}}
\newcommand{\sptwo}{{sp2}}
\newcommand{\href}[2]{#2}
\DeclareMathOperator{\Aut}{Aut}
\DeclareMathOperator{\Ker}{Ker}
\DeclareMathOperator{\SWNT}{SWNT}
\DeclareMathOperator{\Area}{Area}
\DeclareMathOperator{\Vol}{Vol}
\makeatletter
\newenvironment{mynote}[1]{\par\addvspace{17\p@}\small\rm
\trivlist\item[\hskip\labelsep{\bfseries #1}]}
{\endtrivlist\addvspace{6\p@}}
\makeatother
\newcommand{\scalebase}{0.100}
\newcommand{\scalehex}{0.20}    
\newcommand{\scalekite}{0.450}  
\begin{document}
\title*{Construction of Negatively Curved Cubic Carbon Crystals via Standard Realizations}
\author{Hisashi Naito
\\
\vspace{\baselineskip}
Dedicate to Yumiko Naito
}
\authorrunning{Hisashi Naito}
\institute{Hisashi Naito \at Graduate School of Mathematics, Nagoya University, Nagoya Japan, \email{naito@math.nagoya-u.ac.jp}}
\maketitle
\abstract{
  In \cite{Naito-Carbon}, 
  we constructed {\em physically stable} 
  {\sptwo} negatively curved cubic carbon structures
  which reticulate a Schwarz P-like surface.
  The method for constructing such crystal structures 
  is based on the notion of the standard realization of abstract crystal lattices.
  In this paper, 
  we expound on the mathematical method to construct such crystal structures.
}
\section{Introduction}
\label{sec:intro}
In the last few decades, 
there have been many studies about carbon allotropes.
In this paper, we mainly study their geometric structures.
For example, 
it is considered that
C60, graphene sheets, and single wall carbon nanotubes (SWNTs)
span $2$-dimensional surfaces in $\R^3$.
In particular, 
they span the sphere $S^2$, the plane $\R^2$ and cylinders, respectively.
Since the
Gauss curvature of $S^2$ is positive
and
that of $\R^2$ and cylinder is zero, 
C60 spans a positively curved surface 
and
both of graphene sheets and SWNTs span flat surfaces.
\par
It is natural to ask
whether there is an {\sptwo} carbon allotrope which spans a negatively curved surface.
Mackay-Terrones \cite{Mackay} first constructed
such an {\sptwo} carbon crystal structure.
In this paper, we call this ``Mackay-Terrones C192''.
The structure can be placed on the Schwarz P-surface, 
which is a well-known triply periodic minimal surface.
Since minimal surfaces have negative Gauss curvature, 
we may consider Mackay-Terrones C192 to be a negatively curved ``discrete surface''.
Following Mackay-Terrones' work, 
there have been many examples of carbon crystal structures 
which can be considered to be triply periodic {\sptwo} negatively curved discrete surfaces
(cf. \cite{H.Terrones, H.Terrones3, Lenosky, R.Phillips, M-Z.Huang, Townsend, Mackay-Fowler}).
In \cite{Naito-Carbon}, we also constructed such examples, 
and the main features of our method were
mathematical and systematic.
To construct such structures, 
we used the notion of the standard realization of topological crystal (cf. \cite{Kotani-Sunada, Sunada-Book}).
\par
In this paper, 
first, in Section 2, we summarize the combinatorial structures of {\sptwo} carbon allotropes 
and define the notion of ``discrete surfaces'' and their ``total discrete curvature'' .
In Section \ref{sec:Standard}, 
we expound the notion of topological crystals and their standard realizations.
After that, in Section 4, 
we discuss how to construct such structures by using the standard realization of topological crystals.
\par
The main purpose of this paper is 
to construct {\em physically stable} negatively curved {\sptwo} carbon structures.
In particular, we only consider structures that can be placed on a surface whose symmetry is the same as the Schwarz P-surface.
We note that
when we discuss the mathematical structure of {\sptwo} ``carbon'' networks, 
atomic species do not effect the construction of such structures except for searching for physically stable structures by using first principle calculations.
\section{Minimal surfaces and schwarzites}
\label{sec:Minimal}
Minimal surfaces in $\R^3$ are examples of surfaces with negative curvature.
Moreover, there are many examples of triply periodic minimal surfaces.
One well-known example is the Schwarz P-surface \cite{Schoen, Schoen2, Karcher-Polthier, Molnar}  (cf. Figure \ref{fig:schwartz-P}(a)), 
which has following properties:
\begin{enumerate}
\item 
  \label{s:item:1}
  It is a triply periodic minimal surface, 
  whose fundamental domain is cubic,
\item 
  \label{s:item:2}
  It is parameterized by a conformal map from 
  a Riemann surface with genus $3$, 
\item 
  \label{s:item:3}
  It has the same symmetry
  the \net{pcu} net.
  In other words, its space group is \group{Pm\overline{3}m}.
\end{enumerate}
We remark that 
Schwarz D- and G-surfaces also satisfy properties \ref{s:item:1} and \ref{s:item:2}, 
and the following discussion may apply to 
D- and G-surfaces.
However, for simplicity,
we only discuss the Schwarz P-surface by assuming property \ref{s:item:3}.
Here, \emph{cubic} means that 
the period lattice is orthogonal, 
that is to say, the gram matrix of the period lattice is proportional to the identity matrix.
In \cite{Mackay}, 
Mackay-Terrones construct a carbon crystal structure which is placed on 
the Schwarz P-surface (cf. Figure \ref{fig:schwartz-P}(b)).
Since the structure is placed on a triply periodic minimal surface and
it is a crystal structure, 
hence,
it can be considered to be an example of a ``negatively curved carbon crystal''.
Here, we note that structures which are placed on a triply periodic minimal surface are called \emph{schwarzites},
whereas structures which are placed on a positively curved surface are called \emph{fullerenes}.
More precisely, 
a schwarzite is
a trivalent network (an {\sptwo} structure)
that reticulates a triply periodic hyperbolic surface.
As we mention in Remark \ref{remark:euler}, 
such a structure contains rings larger than hexagons.
Later, in \cite{Lenosky}, Lenosky et al.~also constructed similar structures.
There are many works constructing schwarzites
(cf. \cite{H.Terrones, H.Terrones3, M-Z.Huang, R.Phillips, Other2, Other3, Other4, Other5, Other6}).
\begin{figure}
  \centering
  \begin{tabular}{lll}
    (a)&(b) &(c)
    \cr
    \includegraphics[bb=0 0 640 640,scale=0.15]{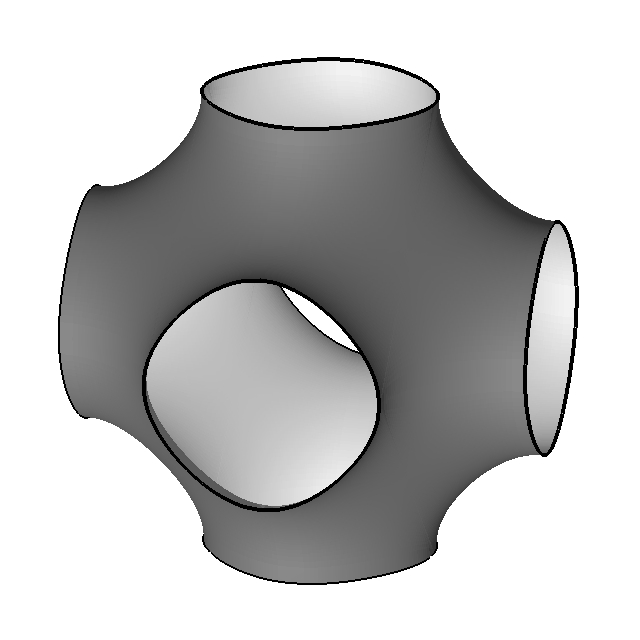}
    &
    \includegraphics[bb=0 0 640 640,scale=0.15]{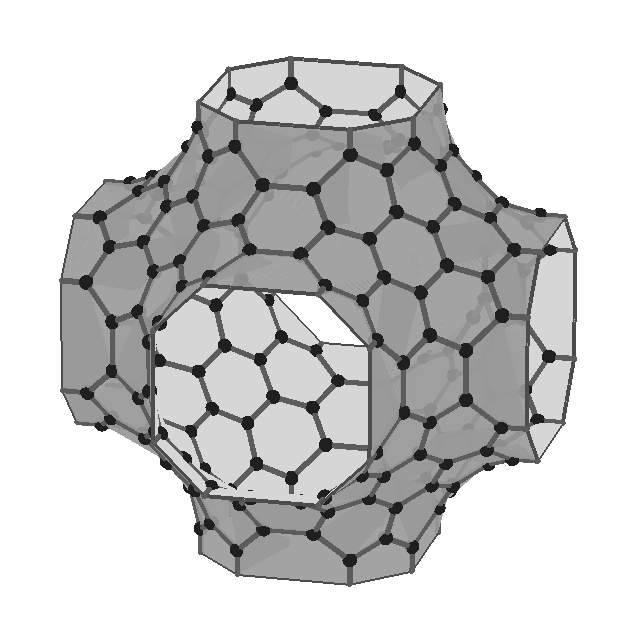}
    &
    \includegraphics[bb=0 0 640 640,scale=0.15]{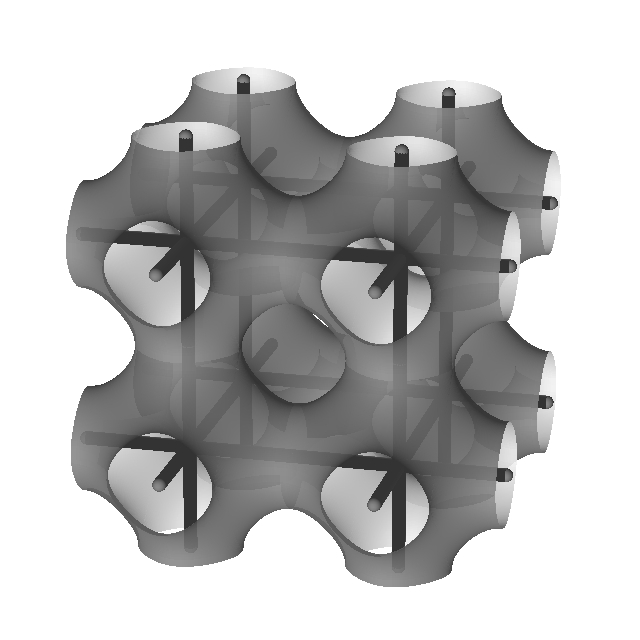}
  \end{tabular}
  \caption{
      (a) 
      Fundamental domain of Schwarz P-surface
      with respect to parallel transformations.
      The total surface is expanded by orthogonal parallel transformations.
      (b)
      Mackay-Terrones C192, 
      embedded on Schwarz P-surface.
      (c)
      Schwarz P-surface and \net{pcu} net.
  }
  \label{fig:schwartz-P}
\end{figure}

\par
On the other hand, 
locally finite graphs are suitable mathematical objects
to consider as molecular/crystal structures.
We consider not only merely abstract graph structures
but also realizations of graphs.
As mentioned in the Introduction, 
since we will study surfaces 
which are spanned by molecular/crystal structures, 
we consider polyhedra constructed by realizations of graphs, 
and call them ``discrete surfaces''.
\begin{definition}[\!{\cite[Section 7.1]{Sunada-Book}}\,]
  \label{definition:realization}
  Let $X = (V, E)$ be a locally finite graph.
  A map $\Phi \colon X \longrightarrow \R^d$ is called a $d$-dimensional \emph{realization} of $X$, 
  identifying $X$ with a 1-dimensional cell complex.
\end{definition}
From this section, 
a realization $X$ of a locally finite graph
denotes a molecular/crystal structure, 
and we only consider $3$-dimensional ones.
Moreover, 
if the structure is periodic with period lattice $\Gamma$, 
we write $X/\Gamma$ as $X$
if there is no confusion.
Hence, we only consider only realizations of a finite graphs
to study molecular/crystal structures.
\par
Now we define the notion of ``discrete surfaces''.
Let $X$ be a realization of a finite graph.
Moreover, we assume each simple closed path of $X$ 
spans a surfaces which does not self-intersect.
Since such a realization is considered to be a polyhedron, 
we call $X$ a \emph{discrete surface}.
In particular, 
if the underlying finite graph of $X$ is of degree $3$, 
we call $X$ a \emph{discrete surface of degree $3$}.
\par
Now,
we consider C60, SWNTs and {\sptwo} schwarzites.
Expressing them by realizations of graphs, 
each underlying graph of $X$ is of degree $3$ (\emph{trivalent}), 
since we consider only {\sptwo} structures.
\par
Let $V(X)$, $E(X)$ and $F(X)$ be the numbers of vertices, edges and faces 
in a discrete surface $X$.
By considering $X$ as a CW complex, 
we may define the genus $g(X)$ of $X$ as $g(X) = \dim H_1(X, \Z)$.
Then, by Euler's theorem (Euler-Poincar\'e theorem), 
\begin{equation}
  \label{eq:euler}
  V(X) - E(X) + F(X) = 2 - 2g(X).
\end{equation}
For a smooth surface $M$, by the Gauss-Bonnet Theorem, 
$2\pi(2 - 2g) = K(M)$, where $K(M)$ is the total curvature of $M$.
Therefore, we define the \emph{total discrete curvature} $K(X)$ of $X$ by 
\begin{displaymath}
  K(X) = V(X) - E(X) + F(X).
\end{displaymath}
Here, we only consider the signature of $K(X)$, 
and we say that $X$ is totally negatively curved if and only if $K(X) < 0$.
Hence, we can call a carbon structure negatively curved
if the realization $X$ of the graph of the structure is negatively curved.
\begin{example}
  C60 is totally positively curved, 
  and the
  total discrete curvatures of SWNTs and graphene sheets are zero.
  Mackay-Terrones C192 is totally negatively curved (cf. Table \ref{table:euler}).
\end{example}
\begin{remark}
  \label{remark:euler}
  For a discrete surface $X$ of degree $3$, 
  we obtain
  \begin{equation}
    \label{eq:VEF}
    V(X) = \frac{k}{3}\sum_k N_k, 
    \quad
    E(X) = \frac{k}{2}\sum_k N_k, 
    \quad
    F(X) = \sum_k N_k, 
  \end{equation}
  where $N_k$ is the number of $k$-gon.
  By using Euler theorem, 
  we, therefore, obtain 
  \begin{equation}
    \label{eq:2:e}
    \sum_k \left(1-\frac{k}{6}\right) N_k = 2 - 2 g(X).
  \end{equation}
  If $X$ satisfies $K(X) < 0$, 
  equality (\ref{eq:2:e}) implies that
  $X$ contains at least one $k$-gon ($k \ge 7$).
\end{remark}
\begin{table}
  \centering
  \caption{
    C60 structure is a truncated icosahedron, 
    For graphene sheets and SWNTs, 
    we take fundamental domain with respect to 
    $\Z^2$- and $\Z$-action, respectively.
    For Mackay-Terrones C192, 
    we take the fundamental domain with respect to
    $\Z^3$-action.
    See Appendix for chiral index of SWNTs.
  }
  \label{table:euler}
  \begin{tabular}{l|r|r|r||r|r}
    & $V(X)$ & $E(X)$ & $F(X)$ & $K(X)$ & $g(X)$ 
    \cr
    \hline
    \hline
    C60 
    & $60$ & $90$ & $32$ & $2$ & $0$ 
    \cr
    \hline
    SWNT, chiral index = $(6, 0)$
    & $24$ & $36$ & $12$ & $0$ & $1$ 
    \cr
    \hline
    SWNT, chiral index = $(6, 3)$
    & $64$ & $126$ & $42$ & $0$ & $1$ 
    \cr
    \hline
    SWNT, chiral index = $(6, 6)$
    & $24$ & $36$ & $12$ & $0$ & $1$ 
    \cr
    \hline
    Mackay-Terrones C192
    & $192$ & $288$ & $102$ & $-4$ & $3$ 
    \cr
  \end{tabular}
\end{table}
\section{Topological crystals and their standard realizations}
\label{sec:Standard}
Scientists use space groups (crystallographic groups) to describe structures of crystals.
A space group represents the symmetry of atoms in the crystal, 
but it does not describe bonds of atoms.
The notion of a crystal lattice or a topological crystal, which was introduced by Kotani-Sunada \cite{Kotani-Sunada}, 
represents placements of atoms and their bonds in a crystal.
The definition of topological crystals is as follows:
\begin{definition}[\!{\cite[Section 6.2]{Sunada-Book}}\,]
  A locally finite graph $X = (V, E)$ is called a \emph{topological crystal} or a \emph{crystal lattice} if and only if
  there exists a finite graph $X_0 = (V_0, E_0)$ and 
  a regular covering map $\pi \colon X \longrightarrow X_0$
  such that the covering transformation group $\Gamma$ of $\pi$ is abelian.
  Moreover, if the rank of $\Gamma \subset H_1(X_0, \Z)$ is $d$, 
  $X$ is called $d$-dimensional.
\end{definition}
A topological crystal $X = (V, E)$ is an abstract structure of a crystal, 
that is to say, each vertex of $X$ represents an atom of a crystal, 
and if $\vv{v}_1$ and $\vv{v}_2 \in V$ are connected by an edge $(\vv{v}_1, \vv{v}_2) \in E$, 
then atoms $\vv{v}_1$ and $\vv{v}_2$ are bonded.
However, the coordinates of vertices are not defined.
To determine, the coordinates of vertices, 
we define a realization of $X$.
\begin{definition}[\!{\cite[Section 7.1]{Sunada-Book}}\,]
  \label{definition:periodicrealization}
  Let $X = (V, E)$ be a $d$-dimensional topological crystal, 
  which is identified with a 1-dimensional cell complex.
  A map $\Phi \colon X \longrightarrow \R^d$ is called a \emph{realization} of $X$.
  Moreover, if there exists an injective homomorphism $\rho \colon \Gamma \longrightarrow \R^d$ such that
  $\Phi(\sigma \vv{v}) = \Phi(\vv{v}) + \rho(\sigma)$ for any $\vv{v} \in V$ and $\sigma \in L$, 
  and $\rho(\Gamma)$ is a lattice subgroup of $\R^d$, 
  then $\Phi$ is called a \emph{periodic} realization.
\end{definition}
Since physical crystals have periodic structures, 
it is natural to consider only periodic realizations.
Figure \ref{fig:hexagonal} shown examples of realizations of the same crystal lattice.
As show in Figure \ref{fig:hexagonal}, 
we would like to select good ones among all realizations of a given crystal lattice.
To do this, 
we define the energy functional with respect to realizations of a given crystal lattice.
\begin{definition}[\!{\cite[Section 7.3]{Sunada-Book}}\,]
  The \emph{energy} of a realization $\Phi$ of a crystal lattice $X$ is defined by
  \begin{displaymath}
    E(\Phi) 
    =
    \frac{1}{2} \sum_{(\vv{v}_1, \vv{v}_2) \in E_0} |\Phi(\vv{v}_1) - \Phi(\vv{v}_2)|^2, 
  \end{displaymath}
  where $X_0$ is the fundamental graph of $X$.
  Moreover,
  $\Phi$ is called a \emph{harmonic} realization
  if $\Phi$ is a critical point of $E$.
\end{definition}
\begin{figure}
  \centering
  \begin{tabular}{llll}
    (a) & (b) & (c) & (d)
    \cr
    \includegraphics[bb=0 0 65 65,scale=1.2,clip=true]{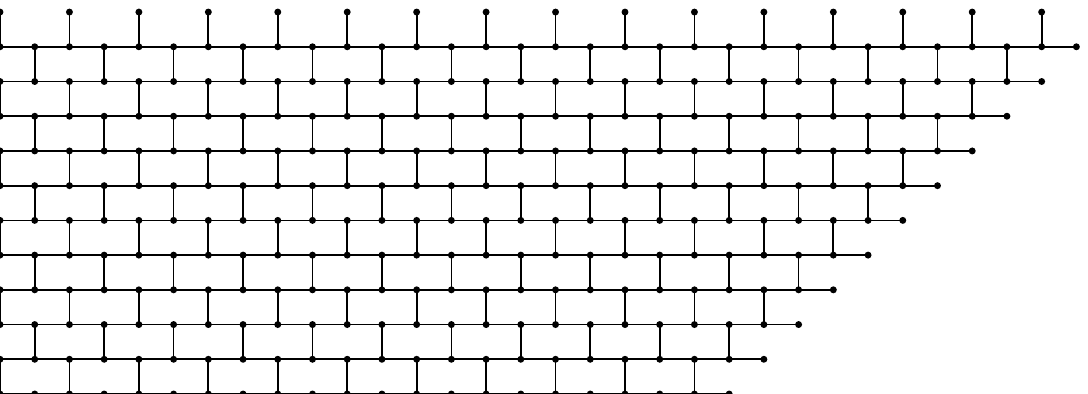}
    &
    \includegraphics[bb=0 0 65 65,scale=1.2,clip=true]{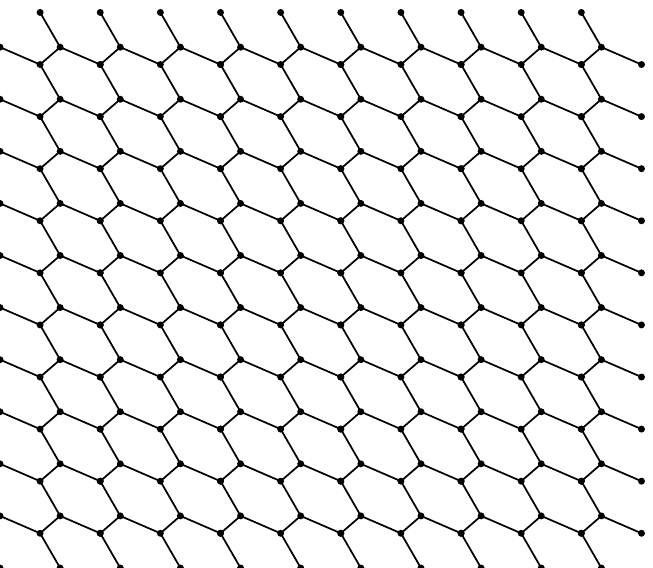}
    &
    \includegraphics[bb=0 0 65 65,scale=1.2,clip=true]{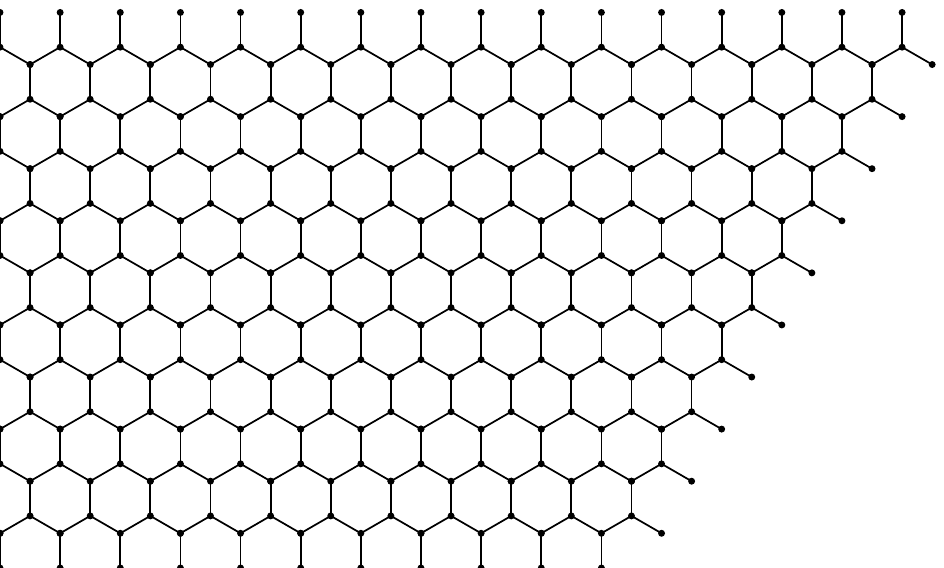}
    &
    \includegraphics[bb=0 0 65 65,scale=1.2]{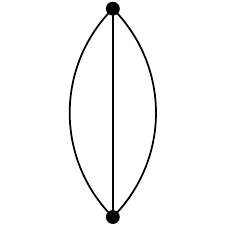}
  \end{tabular}
  \caption{
    Structures (a), (b) and (c) are periodic realizations of the same topological crystals (hexagonal lattice).
    (a) is not a harmonic realization, 
    (b) and (c) are harmonic.
    Moreover (c) is the standard realization of the hexagonal lattice.
    (d) is the fundamental graph of the hexagonal lattice.
  }
  \label{fig:hexagonal}
\end{figure}

It is easy to show that 
$\Phi$ is harmonic if and only if
$\Phi$ satisfies
\begin{equation}
  \label{eq:3:0}
  \sum_{(\vv{v}, \vv{v}_1) \in E_0} (\Phi(\vv{v}) - \Phi(\vv{v}_1)) = 0, 
\end{equation}
for all $\vv{v} \in X_0$.
Note that
1) the left-hand side of (\ref{eq:3:0}) is the discrete Laplacian of $(\Phi(V_0), E_0)$;
2) $\Phi$ is harmonic if and only if 
all vertices of $\Phi(X)$ satisfy mechanical equilibrium;
3) structures in Figure \ref{fig:hexagonal}(b) and Figure \ref{fig:hexagonal}(c) are harmonic.
However, we consider the structure in Figure \ref{fig:hexagonal}(c)
as the most symmetric object among all realizations of $X$.
To select the most symmetric object by using the variational method, 
we also consider variations with respect to lattices $\Gamma$.
\begin{definition}[\!{\cite[Section 7.4]{Sunada-Book}}\,]
  Let $X$ be a $d$-dimensional crystal lattice 
  with the fundamental graph $X_0$.
  Moreover,
  let $\Gamma$ be the covering transformation group.
  A realization $\Phi \colon X \longrightarrow \R^d$ 
  is called \emph{standard}
  if and only if
  $\Phi$ is a critical point of $E$ with respect to 
  variations $\Phi$ and $\Gamma$ subject to $\Vol(\Gamma) = 1$.
\end{definition}
\begin{theorem}[\!{\cite{Kotani-Sunada}}\,]
  For any topological crystal $X$, 
  there exists a standard realization $\Phi$ of $X$.
  Moreover $\Phi$ is unique up to scaling and parallel translations.
\end{theorem}
\begin{remark}
  In \cite[Theorem 4]{Delgado-Friedrichs-2}, 
  Delgado-Friedrichs considered equilibrium placements 
  $\Phi \colon X \to \R^d$ 
  of 
  $d$-periodic graphs $X$, 
  and proved the unique (up to affine translations) 
  existence of 
  the equilibrium placement for any periodic graph
  (See also \cite{Delgado-Friedrichs-1}).
  The notion of equilibrium placements 
  corresponds to the notion of harmonic realizations.
  Moreover, Delgado-Friedrichs also proved that, 
  for any $d$-periodic graph $X$, 
  there exists a unique equilibrium placement $\Phi$ such that, 
  for every $\gamma \in \Gamma \subset \Aut(X)$,
  an isometry $\gamma^* \colon \R^d \to \R^d$ associated to $\gamma$ with respect to $\Phi$ exists (\cite[Theorem 11]{Delgado-Friedrichs-2}).
  This equilibrium placement corresponds to the standard realization.
\end{remark}
\begin{example}
  For the hexagonal lattice (Figure \ref{fig:hexagonal}), 
  all lattice $\Gamma$ are parameterized by the angle $\theta$ between 
  period vectors.
  The energy of all harmonic realizations attains its minimum at 
  $\theta = \pi/3$, 
  which implies that
  the standard realization of the hexagonal lattice is
  Figure \ref{fig:hexagonal}(c).
\end{example}
\begin{example}
  There is a general method
  of calculating the standard realization of 
  a given topological crystal (cf.~\cite{Sunada-Book, Naito-AMS}).
  The diamond crystal is obtained by this procedure from 
  the fundamental graph in Figure \ref{fig:3d}(a).
  The \net{srs} network, which has a deep relationship with the Gyroid surface,  
  is also obtained by this method.
  Moreover, Sunada \cite{Sunada-Notice} shows that the only networks which satisfy the ``strongly isotropic property'' are diamond (\net{dia}) and \net{srs}.
  For the relation of the \net{srs} network and the Gyroid surface, 
  see \cite{Hyde-2, Hyde-3}.
  For the definition of the strongly isotropic property, 
  see \cite[p.212]{Sunada-Notice}, 
  and 
  for the explicit figure of $K_4$, 
  see \cite{Naito-AMS}, \cite[Figure 5]{Hyde-3}.
  Here, we note that 
  we show physical meta-stability of carbon $K_4$ structure in \cite{K4}.
\end{example}
\begin{figure}
  \centering
  \sidecaption
  \begin{tabular}[b]{ll}
    (a)& (b)
    \cr
    \includegraphics[bb=0 0 65 65,scale=1.5]{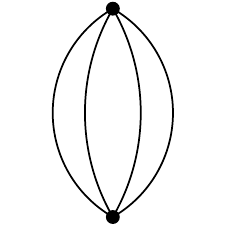}
    &
    \includegraphics[bb=0 0 65 65,scale=1.5]{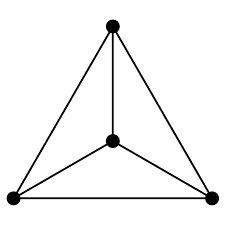}
  \end{tabular}
  \caption{
      (a) 
      The fundamental graph of the diamond lattice.
      The diamond lattice is obtained by 
      the standard realization of the maximal abelian covering of this graph.
      (b)
      The fundamental graph ($K_4$ graph) of the \net{srs} net (the $K_4$ lattice).
      The \net{srs} net is obtained by 
      the standard realization of the maximal abelian covering of $K_4$ graph.
  }
  \label{fig:3d}
\end{figure}

The most important property of the standard realization is as follows:
\begin{theorem}[\!{\cite[Theorem 1]{Sunada-Notice}}\,]
  \label{claim:sunada}
  Let $X$ be a $d$-dimensional crystal lattice, 
  and $\Phi$ a realization of $X$.
  Then there exists a homomorphism $T \colon \Aut(X) \longrightarrow M(d)$
  such that $\Phi(g\vv{x}) = T(g) \Phi(\vv{x})$ for any 
  $\vv{x} \in X$ and $g \in \Aut(X)$, 
  where $\Aut(X)$ is the automorphism group of $X$ and
  $M(d)$ is the motion group of\/ $\R^d$.
\end{theorem}
Theorem \ref{claim:sunada}
implies that
the standard realization is
the most symmetric among all realizations.
%
\section{Construction of negatively curved carbon crystals}
\label{sec:Carbon}
Our main aim is to construct examples of physically stable {\sptwo} negatively curved carbon structures.
In particular, we construct such structures with 
the same symmetry as \net{pcu}.
For the 3-net \net{pcu}, see Delgado-Friedrichs et al.~\cite{D-Friedrichs} and Hyde et al.~\cite{Hyde-3}.
In the following, we abbreviate the symmetry as {\em cubic symmetry}.
In this section, we use the word {\em network}, 
which means a realization of graphs.
\subsection{Construction of topological crystals}
\label{sec:Carbon:topology}
To construct network (graph) structures with cubic symmetry, 
first we construct them in the \emph{hexagonal region}, 
then extend to the fundamental region (unit cell) of crystals
(see Figure \ref{fig:schwartz-P}(a) and Figure \ref{fig:8-2-1}(a)).
Since the hexagonal region has reflective symmetry of order $3$, 
we only consider networks with such symmetry.
We call the fundamental region of such symmetry the \emph{kite-region}
(see Figure \ref{fig:8-2-1}(b)).
Using the orbifold notation, 
the kite-region is the $\ast2223$ orbifold, 
and the hexagonal region is the $\ast 222222$ orbifold in $\H^2$.
For the orbifold notation, see \cite{Conway-Huson, epinet}
\begin{figure}
  \centering
  \begin{tabular}[b]{lll}
    (a) & (b) & (c)
    \cr
    \includegraphics[height=100pt]{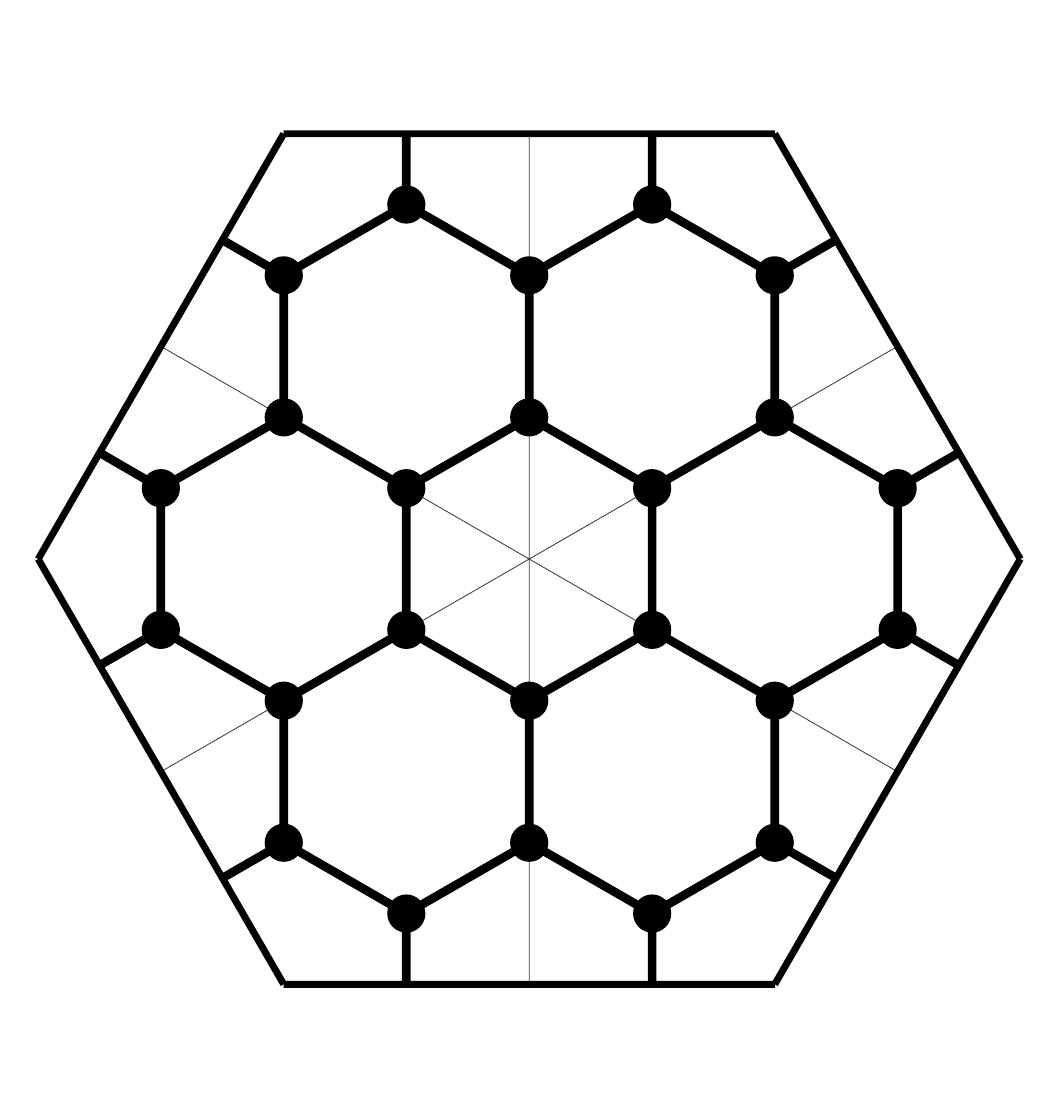}
    &
    \includegraphics[height=100pt]{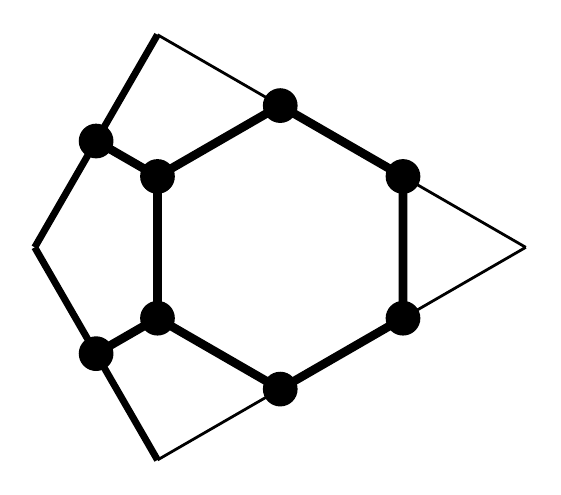}
    &
    \includegraphics[height=100pt]{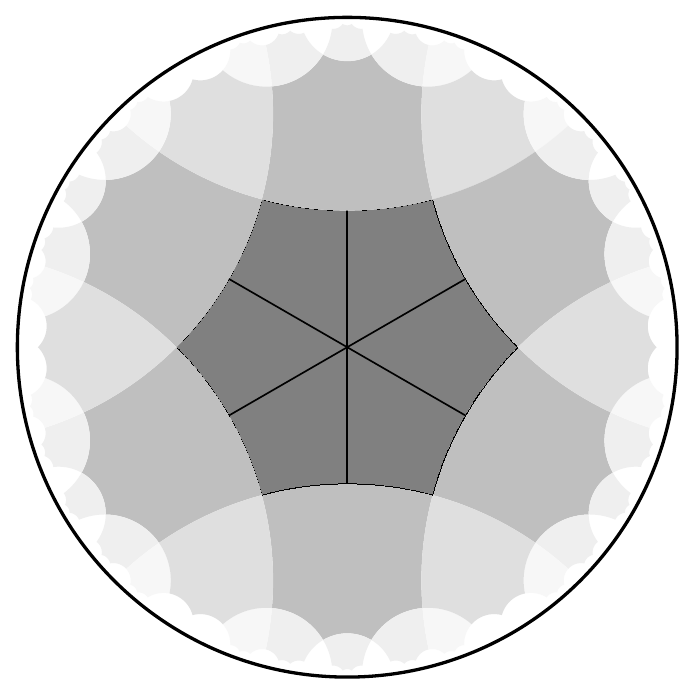}
    \cr
  \end{tabular}
  \caption{
    (a)
    The combinatorial structure of Mackay-Terrones C192
    in the hexagonal region $\ast222222$).
    Mackay-Terrones C192 is obtained by 
    reflecting this structure by group action of $\ast222222$.
    (b)
    The combinatorial structure of Mackay-Terrones C192
    in the kite-region ($\ast2223$).
    Structure (a) is obtained by 
    reflecting the structure in the kite-region with respect to 
    lines through the right vertex 
    (the vertex with $\pi/3$ symmetry).
    (c)
    An image of the hexagonal region
    in $\H^2$ (See \cite[Figure 7(c)]{Hyde-4}).
  }
  \label{fig:8-2-1}
\end{figure}

\par
As the first step to construct such networks, 
we construct networks in the kite-region satisfying the following properties:
\begin{enumerate}
\item 
  \label{e:network:1}
  Any inner vertex is of degree 3, 
\item 
  \label{e:network:2}
  Any vertex on the boundary is joined with the two neighbouring vertices on the boundary, 
  or with an inner vertex and not with both neighbouring vertices on the boundary, 
\item 
  \label{e:network:3}
  A network is planar and connected, 
  and there are at least four vertices on the boundary, 
\item 
  \label{e:network:4}
  A network does not have a consecutive sequence of odd vertices on the boundary, 
\item 
  \label{e:network:5}
  A network is triangle-free.
\end{enumerate}
Condition \ref{e:network:1} corresponds to consideration of {\sptwo} crystals.
By conditions \ref{e:network:2}, \ref{e:network:3} and \ref{e:network:4}.
we may extend the network to the hexagonal region.
Condition \ref{e:network:5} implies avoiding $3$-rings in structures.
It is easy to prove that networks with the above conditions must have an even number of vertices.
Networks with $6$ and $8$ vertices 
in the kite-region
are classified in \cite[Figure 2]{Naito-Carbon}.
Since the kite-region is homeomorphic to the disk, 
we may construct such networks in the disk region.
\par
The second step is to extend networks in the kite-region to the hexagonal region by reflections.
In this step, 
if there exist degree $2$ vertices on reflection boundaries, 
we delete them.
By these procedures, 
we may construct a network in the hexagonal region.
\par
Finally, 
by extending the network in the hexagonal region to the fundamental region by parallel transformations, 
we obtain a required network.
We note that the above properties are not sufficient conditions
to obtain trivalent networks which reticulate
a Schwarz P-like surface.
Constructing networks in the hexagonal region, 
and patching them in a suitable manner, 
we may obtain networks on D- and/or G-like surface.
However, as the aim of this note is 
to construct examples of networks on a P-like surface, 
we consider the hexagonal region as 
a hexagonal face of the trucated octahedron, 
then we obtain a required trivalent network 
by extending in the above manner.
\subsection{Construction of standard realizations}
\label{sec:Carbon:realization}
Let $X_0 = (V_0, E_0)$ be a network constructed as in Section \ref{sec:Carbon:topology}.
Since $X_0$ has 
cubic symmetry, 
we may extend $X_0$ to a $3$-dimensional toplogical crystal $X = (V, E)$.
Our purpose in this section is
to construct the standard realization of $X$, 
and to prove that the realization has cubic symmetry.
That is to say, 
translation vectors of the standard realization are orthogonal.
\par
First we let $\Phi \colon X \longrightarrow \R^3$ be a realization 
with $\Gamma = X/X_0$, 
and $\{\vv{e}_x, \vv{e}_y, \vv{e}_z\}$ be the basis of $\R^3$ with $\det(\vv{e}_x \vv{e}_y \vv{e}_z) = 1$.
By using this realization, 
we may find coordinates of vertices $\vv{x}_i = \Phi(\vv{v}_i) \in \R^3$ for $\vv{v}_i \in V$.
To calculate standard realizations, 
we should define the energy of realizations.
If $\vv{v}_i \in V_0$, then 
edges $(\vv{v}_i, \vv{v}_j)$ satisfy either
$\vv{v}_j \in V_0$, $\vv{v}_j \in V_0 + \vv{e}_\alpha$ or $\vv{v}_j \in V_0 - \vv{e}_\alpha$
for some $\alpha \in \{x, y, z\}$, 
hence, we write that $(\vv{v}_i, \vv{v}_j)$ is in either $E_0$, $E_{(0,+)}$ or $E_{(0,-)}$.
Under this notation, 
the energy of $\Phi$ may be written as
\begin{equation}
  \label{eq:4:1}
  \begin{aligned}
    E
    = 
    \frac{1}{2}
    &\left(
      \sum_{(\vv{v}_i, \vv{v}_j) \in E_0}
      |\vv{x}_i - \vv{x}_j|^2
    \right.
    \\
    &\left.
      +
      \sum_{(\vv{v}_i, \vv{v}_j) \in E_{(0,+)}}
      |\vv{x}_i + \vv{e}_\alpha - \vv{x}_j|^2
      +
      \sum_{(\vv{v}_i, \vv{v}_j) \in E_{(0,-)}}
      |\vv{x}_i - \vv{e}_\alpha - \vv{x}_j|^2
    \right).
  \end{aligned}
\end{equation}
Since the standard realization is the critical point of $E$ with variations with respect to variables $\{\vv{x}_i\}$ and $\{\vv{e}_\alpha\}$ 
subject to $\det(\vv{e}_x \vv{e}_y \vv{e}_z) = 1$, 
then to obtain the standard realization, 
using Lagrangian multiplier, 
we must solve the following equations:
\begin{equation}
  \label{eq:4:2}
  \frac{\partial E}{\partial \vv{x}_i} = 0,
\end{equation}
and 
\begin{equation}
  \label{eq:4:3}
  \frac{\partial}{\partial \vv{e}_\alpha}\left(E - \lambda \det(\vv{e})\right) = 0.
\end{equation}

\begin{proposition}[\!{\cite[Theorem 2 of Supplementary materials]{Naito-Carbon}}\,]
  The linear system {\upshape (\ref{eq:4:2})} is solvable.
\end{proposition}
\begin{proof}
  The left hand side of equation (\ref{eq:4:2}) is
  \begin{equation}
    \label{eq:4:4}
    \begin{aligned}
      \frac{\partial E}{\partial \vv{x}_i}
      &=
      \vv{x}_{j_1} + \vv{x}_{j_2} + \vv{x}_{j_3} - 3 \vv{x}_j
      + \vv{b}_j
    \end{aligned}
  \end{equation}
  where $\{\vv{v}_{j_1}, \vv{v}_{j_2}, \vv{v}_{j_3}\}$ are vertices adjacent to $\vv{v}_j$, 
  and $\vv{b} = (\vv{b}_i)$ is defined by
  \begin{displaymath}
    \vv{b}_i
    =
    \left\{
      \begin{alignedat}{3}
        &+\vv{e}_\alpha &\quad 
        &\text{if }
        (\vv{v}_j, \vv{v}_{j\ell}) \in E_{(\alpha, +)}
        \text{ for some } \ell \in \{1, 2, 3\}, 
        \alpha \in \{x, y, z\}, \\
        &-\vv{e}_\alpha &\quad 
        &\text{if }
        (\vv{v}_j, \vv{v}_{j\ell}) \in E_{(\alpha, -)}
        \text{ for some } \ell \in \{1, 2, 3\}, 
        \alpha \in \{x, y, z\}, \\
        &0 &\quad &\text{otherwise}
        \\
      \end{alignedat}
    \right.
  \end{displaymath}
  That is to say, 
  the equation (\ref{eq:4:2}) is written as
  \begin{equation}
    \label{eq:4:6}
    \Delta_\M{G} \vv{x} = \vv{b}, 
  \end{equation}
  where $\Delta_\M{G} = \M{A} - 3\M{I}$ is the discrete Laplacian of $G = X/\Gamma$, 
  and $A$ is the adjacency matrix of $G$.
  Since the discrete Laplacian of a connected graph has only a $1$-dimensional kernel, 
  and the kernel is spanned by $(1, \ldots, 1)$.
  Since $\sum \vv{b}_i = 0$, we obtain 
  that 
  $\vv{b}$ is perpendicular to $\Ker \Delta_\M{G}$.
  hence, equation (\ref{eq:4:2}) is solvable.
  \qed
\end{proof}
\begin{theorem}[\!{\cite[Theorem 3 of Supplementary materials]{Naito-Carbon}}\,]
  \label{claim:standard:cubic}
  The gram matrix of the lattice $\{\vv{e}_x, \vv{e}_y, \vv{e}_z\}$ which gives the solution of {\upshape (\ref{eq:4:2})} and {\upshape (\ref{eq:4:3})}
  is the identity matrix.
\end{theorem}
\begin{proof}
  Let $\{\vv{x}_i\}$ be the solution of (\ref{eq:4:2}) with 
  $\sum \vv{x}_i = 0$, 
  and $t_\alpha$ be the reflection with respect to the plane with normal vector $\vv{e}_\alpha$.
  Moreover,
  let $T$ be the group generated by $\{t_x, t_y, t_z\}$.
  Since the energy $E$ is invariant under the action of $T$, 
  the solution $x$ satisfies:
  \begin{enumerate}
  \item 
    \label{e:L:1}
    If $T_x(\vv{v}_i) = \vv{v}_j$, then 
    $\vv{x}_i = K_{i,x}\vv{e}_x + K_{i,y}\vv{e}_y + K_{i,z}\vv{e}_z$
    and 
    $\vv{x}_j = -K_{i,x}\vv{e}_x + K_{i,y}\vv{e}_y + K_{i,z}\vv{e}_z$.
  \item 
    \label{e:L:2}
    $(\vv{v}_i, \vv{v}_j) \in E_{(\alpha, +)}$ if and only if 
    $(\vv{v}_j, \vv{v}_i) \in E_{(\alpha, -)}$.
  \end{enumerate}
  Therefore we obtain
  \begin{displaymath}
    \frac{\partial E}{\partial \vv{e}_\alpha}
    =
    \sum_{(\vv{v}_i, \vv{v}_j)\in E_{(\alpha, +)}} (\vv{x}_j + \vv{e}_\alpha - \vv{x}_i)
    -
    \sum_{(\vv{v}_i, \vv{v}_j)\in E_{(\alpha, -)}} (\vv{x}_j - \vv{e}_\alpha - \vv{x}_i)
    =
    K_\alpha \vv{e}_\alpha.
  \end{displaymath}
  Since the action exchanging $\vv{e}_x$ and $\vv{e}_y$ belongs to $T$, 
  we may obtain $K := K_x = K_y = K_z$ and $K \not= 0$.
  By \ref{e:L:2}, we obtain $K\vv{e}_\alpha = \lambda \vv{e}_\beta \times \vv{e}_\gamma$ 
  with $\epsilon_{\alpha\beta\gamma} = \epsilon_{xyz}$.
  Hence, we obtain that $K\inner{\vv{e}_\alpha}{\vv{e}_\beta} = \lambda \delta_{\alpha\beta}$.
  \qed
\end{proof}

\begin{remark}
  Theorem \ref{claim:standard:cubic} is also obtained from Theorem \ref{claim:sunada}.
  By Theorem \ref{claim:sunada}, 
  the standard realization has maximal symmetry, 
  and the symmetry of the realized crystal must be same as
  that of the topological crystal.
  Since the topological crystal, which we are considering, 
  is invariant under triply periodicity and 
  cubic
  group action, 
  therefore, 
  such actions extend to actions $M(d)$, 
  Hence, the gram matrix of $\{\vv{e}_\alpha\}$ should be proportional to the identity matrix.
\end{remark}
\subsection{Construction of stable configurations}
By using Theorem \ref{claim:standard:cubic}, 
we obtain candidates for negatively curved cubic crystals.
However, distances of neighbouring atoms in such structures are not almost the same.
By physico-chemical considerations, 
the distances should be almost the same for physically stable configurations.
To find stable configurations with respect to binding energies, 
we perturb coordinates of atoms, 
then find stable ones, 
which we call \emph{relaxed configuration}, 
by using first principle calculations.
For this purpose, 
setting the standard realization as the initial configuration, 
perturbing positions of atoms, 
and 
calculating binding energies of carbon structures, 
we obtain stable configurations.
By using the above method, 
we search for physically stable configuration, 
and 
we obtain negatively curved cubic {\sptwo} carbon crystals as examples of schwarzites which reticulate a Schwarz P-like surface.
\begin{result}
  We obtain 
  four physically
  stable structures:
  \name{6-1-1-P} {\upshape (C176)}, 
  \name{6-1-2-P} {\upshape (C152)}, 
  \name{6-1-3-P} {\upshape (C152)}
  and 
  \name{8-4-2-P} {\upshape (C168)}.
  Relaxed configurations of them are
  physically stable.
  Moreover,
  \name{6-1-1-P} {\upshape (C176)}, 
  \name{6-1-2-P} {\upshape (C152)}, 
  \name{6-1-3-P} {\upshape (C152)}
  are metal and 
  \name{8-4-2-P} {\upshape (C168)} is a semi-conductor.
\end{result}
Our structures are illustrated in Figure \ref{fig:result}, 
and Table \ref{table:polygons} lists their basic properties.
For their energy bands and phonon spectrum of them, 
see \cite[Figures 4-8]{Naito-Carbon}.
\begin{remark}
  To calculate relaxed configuration, 
  the only information we use is the coordinates of carbon atoms.
  In Figure \ref{fig:result}(d), 
  we draw edges between an atom and the nearest three atoms.
  As a result, 
  each relaxed configuration has the same topological structure 
  as the original standard realization.
\end{remark}
Hence, our method to obtain cubic schwarzites is effective.
\begin{remark}
  In \cite{Hyde-4}, Ramsden-Robins-Hyde discuss 
  3-dimensional Euclidean nets (3-net).
  Our examples are also 3-nets, and 
  there are many examples which reticulate such surfaces.
  However, we construct, in particular, physically stable
  {\sptwo} (trivalent) examples.
\end{remark}
\begin{remark}
  We note that the structure \name{6-1-2-P} (C152) is 
  almost the same as with C152 in Park et al.~\cite{Park}.
\end{remark}
\begin{table}
  \centering
  \caption{
    Number of $k$-gon $N_k$ of our structures.
    Each structure lies on a negatively curved surface, 
    and 
    has $g = 3$.
    Hence, all of these structures are negatively curved.
    Obviously, each structure satisfies (\ref{eq:2:e}) for $g = 3$.
  }
  \label{table:polygons}
  \begin{tabular}{l|r|r|r|r||r|r|r|r|r}
    schwarzites
    & $N_5$ & $N_6$ & $N_7$ & $N_8$ & 
    \multicolumn{1}{|c}{$V(X)$} & 
    \multicolumn{1}{|c}{$E(X)$} & 
    \multicolumn{1}{|c}{$F(X)$} & 
    \multicolumn{1}{|c}{$K(X)$} & 
    \multicolumn{1}{|c}{$g(X)$} 
    \cr
    \hline
    \hline
    \name{6-1-1-P} (C176)
    & $0$ & $60$ & $24$ & $0$ & 
    $176$ & $264$ & $84$ & $-4$ & $3$
    \cr
    \hline
    \name{6-1-2-P} (C152)
    & $0$ & $60$ & $24$ & $0$ & 
    $152$ & $240$ & $84$ & $-4$ & $3$
    \cr
    \hline
    \name{6-1-3-P} (C152)
    & $24$ & $12$ & $24$ & $12$ & 
    $152$ & $228$ & $72$ & $-4$ & $3$
    \cr
    \hline
    \name{8-4-2-P} (C168)
    & $0$ & $68$ & $0$ & $12$ & 
    $168$ & $252$ & $80$ & $-4$ & $3$
    \cr
    \hline
    Mackay-Terrones C192
    & $0$ & $90$ & $0$ & $12$ & 
    $192$ & $288$ & $102$ & $-4$ & $3$
    \cr
  \end{tabular}
\end{table}
\begin{figure}
  \centering
  \begin{tabular}{lll}
    (a) \name{6-1}
    & 
    (b) \name{8-2}
    & 
    (c) \name{8-4}
    \cr
    \includegraphics[scale=0.65]{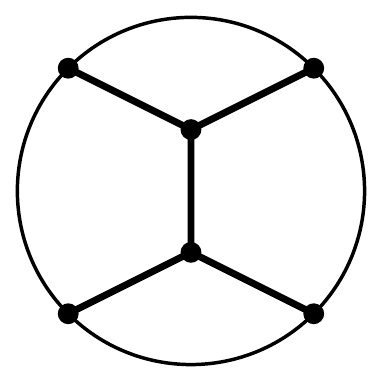}
    &
    \includegraphics[scale=0.65]{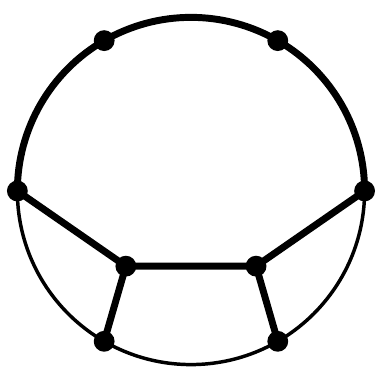}
    &
    \includegraphics[scale=0.65]{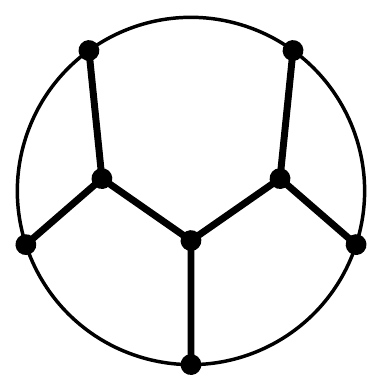}
    \cr
  \end{tabular}
  \caption{
    Networks in circle region.
    See \cite[Figure 2]{Naito-Carbon}
    for other networks with $6$- and $8$-vertices.
    Mackay-Terrones C192 can be constructed from \name{8-2}.
    }
  \label{fig:circle}
\end{figure}

\begin{figure}
  \centering
  \begin{tabular}{c|c|c|c|c}
    &\name{6-1-1-P}
    &\name{6-1-2-P}
    &\name{6-1-3-P}
    &\name{8-4-2-P}
    \cr
    \hline
    \raisebox{30pt}{(a)}
    &
    \includegraphics[bb=0 0 161 142,scale=\scalekite]{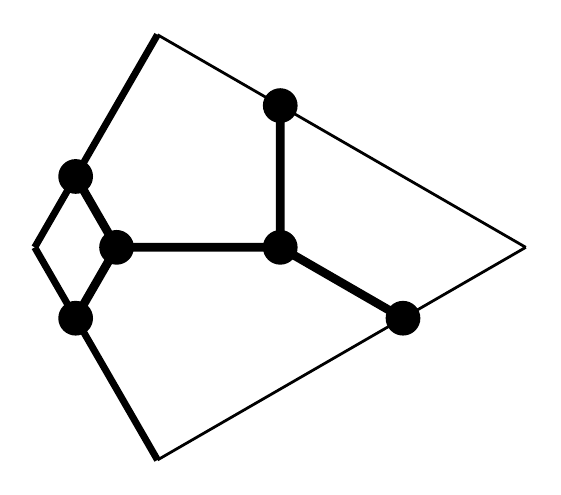}
    &
    \includegraphics[bb=0 0 161 142,scale=\scalekite]{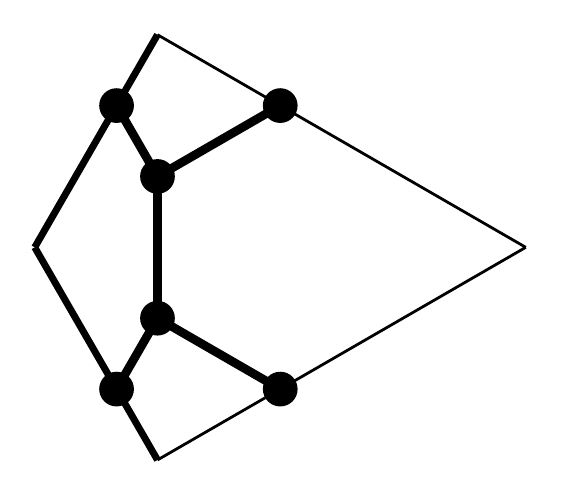}
    &
    \includegraphics[bb=0 0 161 142,scale=\scalekite]{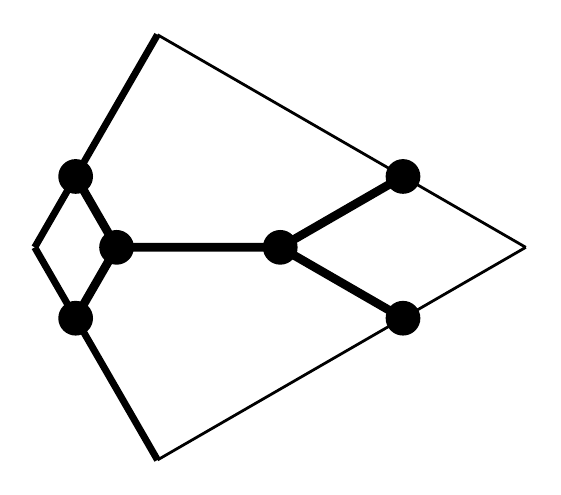}
    &
    \includegraphics[bb=0 0 161 142,scale=\scalekite]{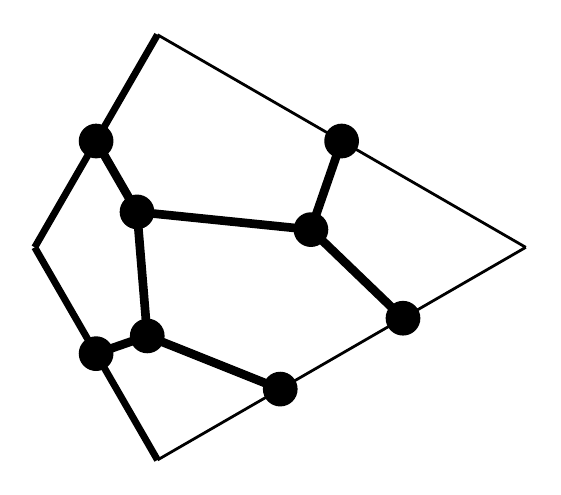}
    \cr
    \hline
    \vspace{0.01pt} & & & 
    \cr
    \raisebox{30pt}{(b)}
    &
    \includegraphics[bb=0 0 305 322,scale=\scalehex]{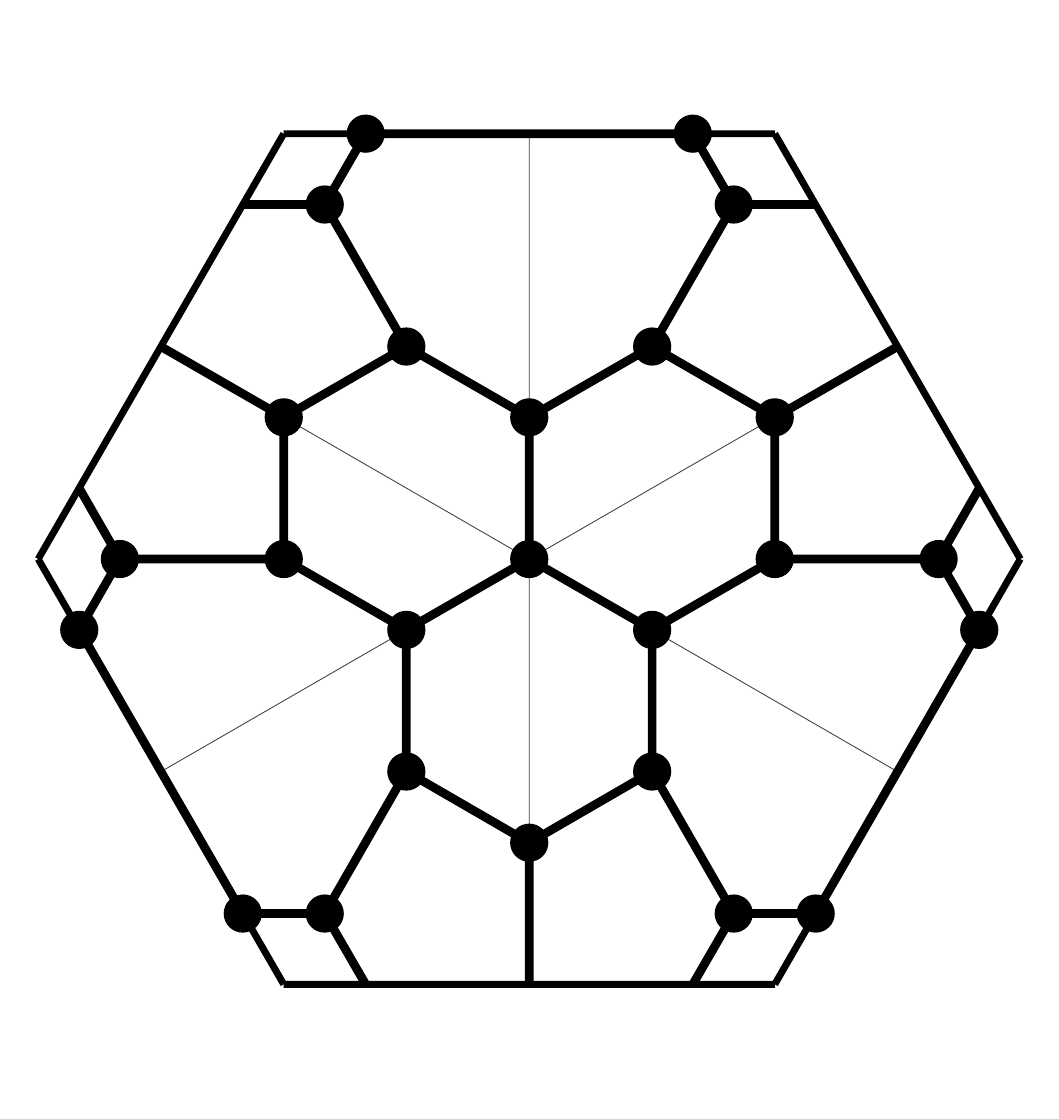}
    &
    \includegraphics[bb=0 0 305 322,scale=\scalehex]{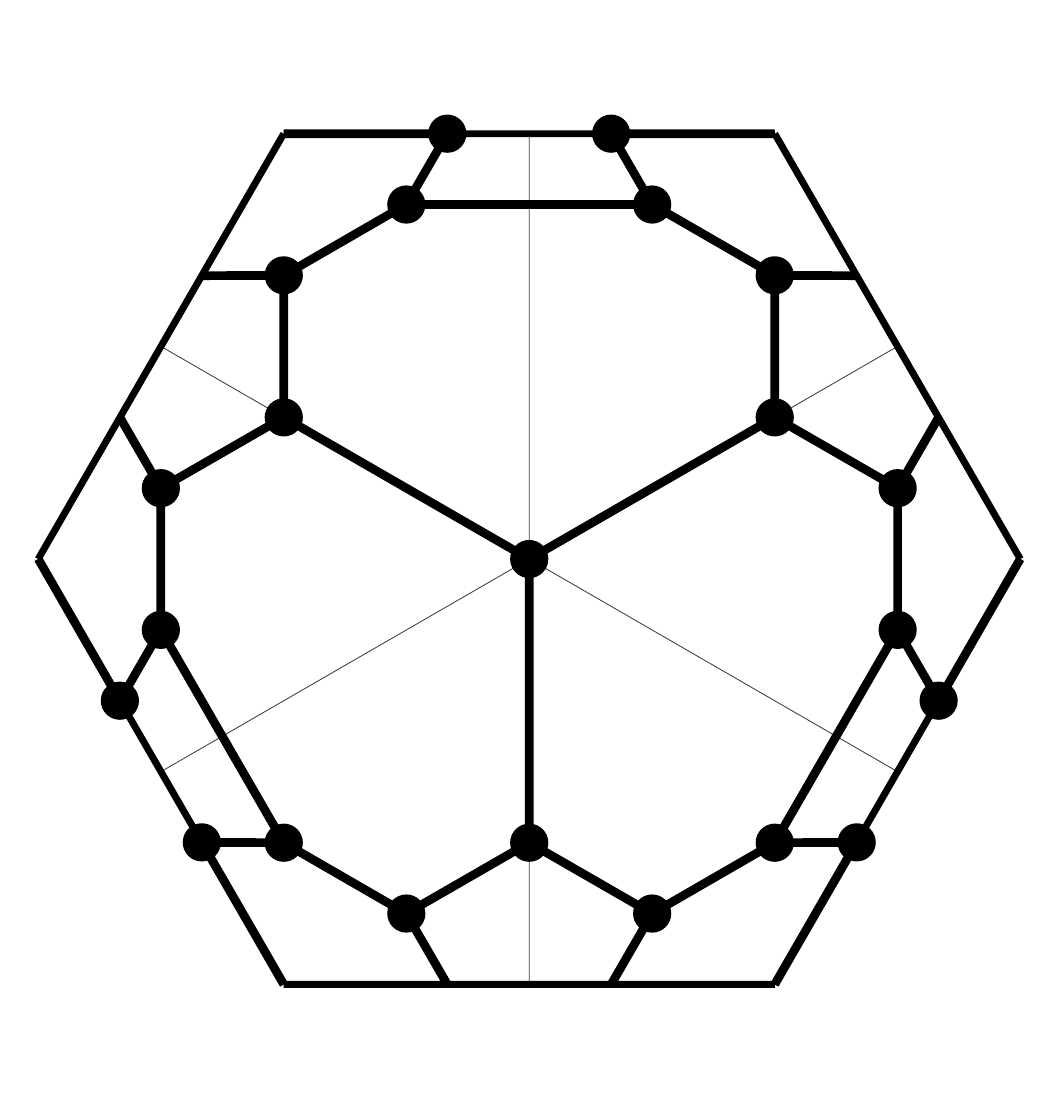}
    &
    \includegraphics[bb=0 0 305 322,scale=\scalehex]{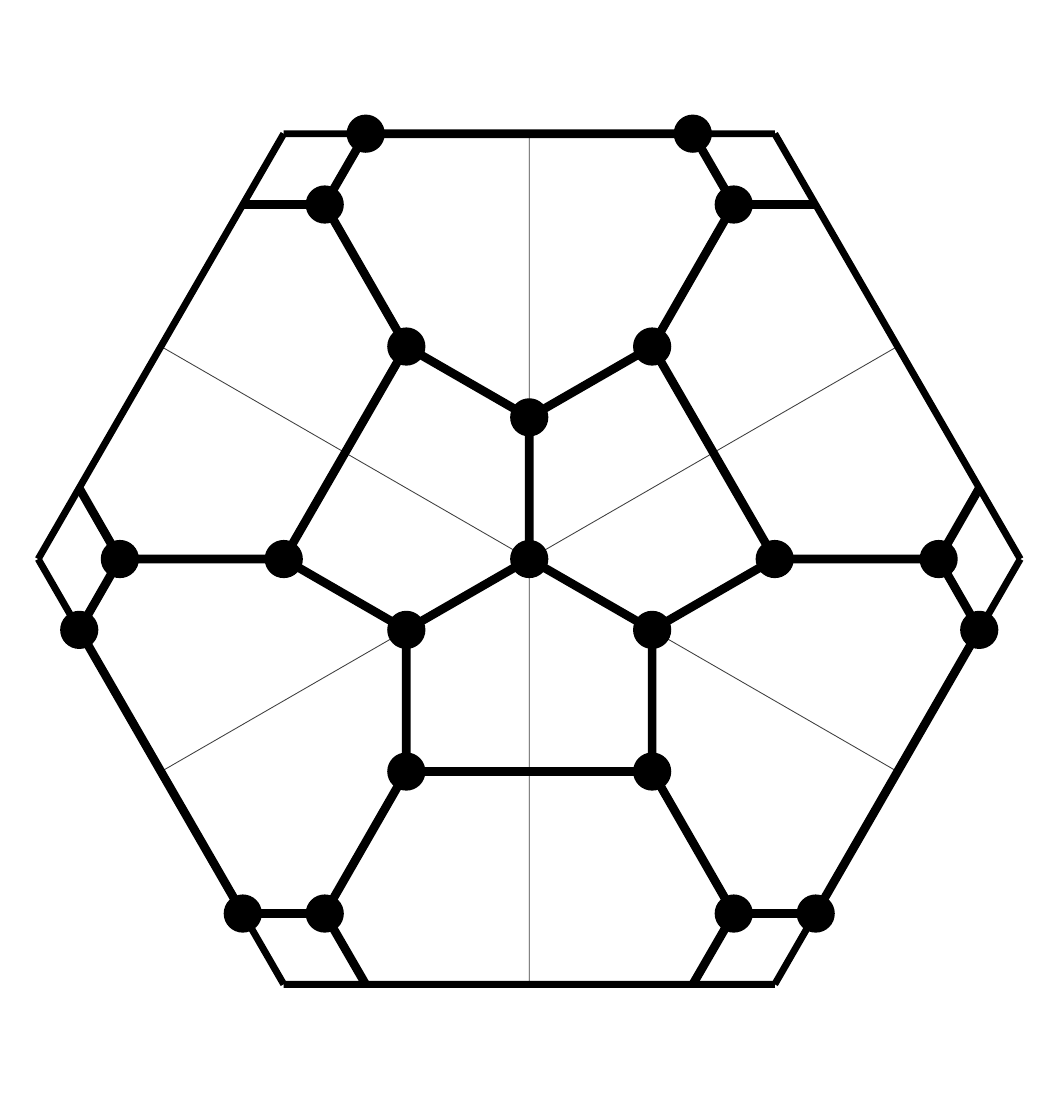}
    &
    \includegraphics[bb=0 0 305 322,scale=\scalehex]{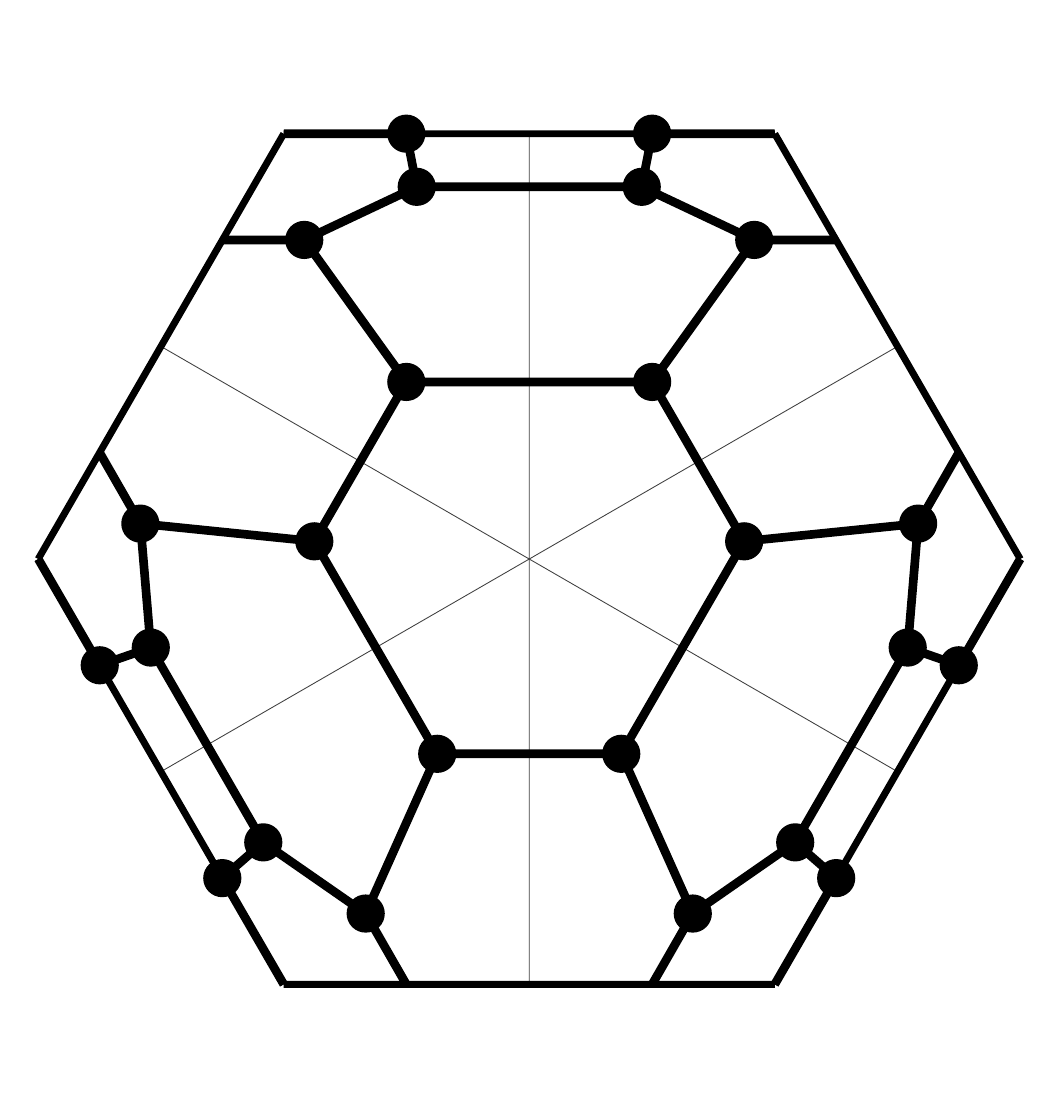}
    \cr
    \hline
    \vspace{0.01pt} & & & 
    \cr
    (c)
    &
    \includegraphics[bb=0 0 640 640,scale=\scalebase]{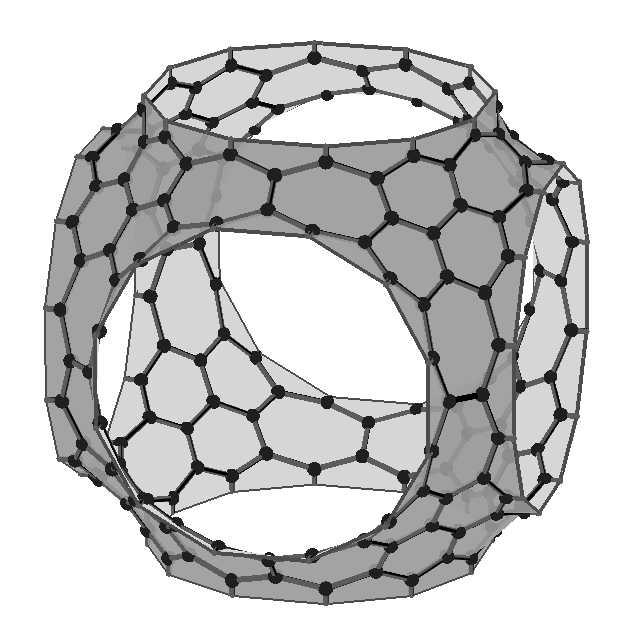}
    &
    \includegraphics[bb=0 0 640 640,scale=\scalebase]{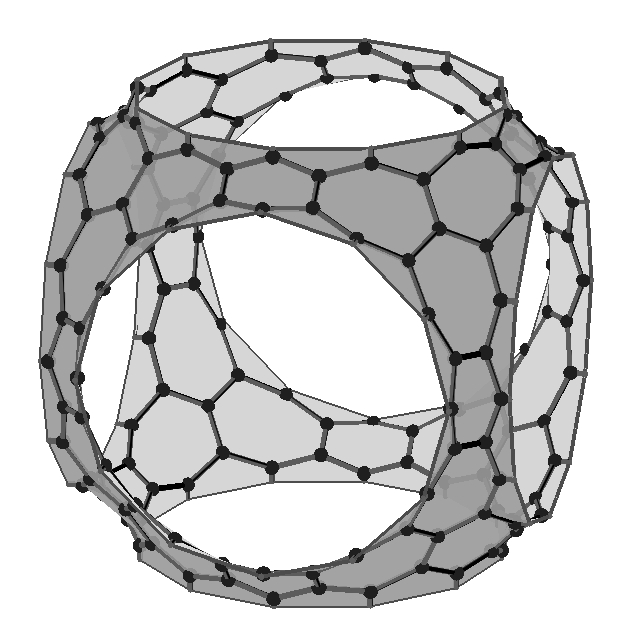}
    &
    \includegraphics[bb=0 0 640 640,scale=\scalebase]{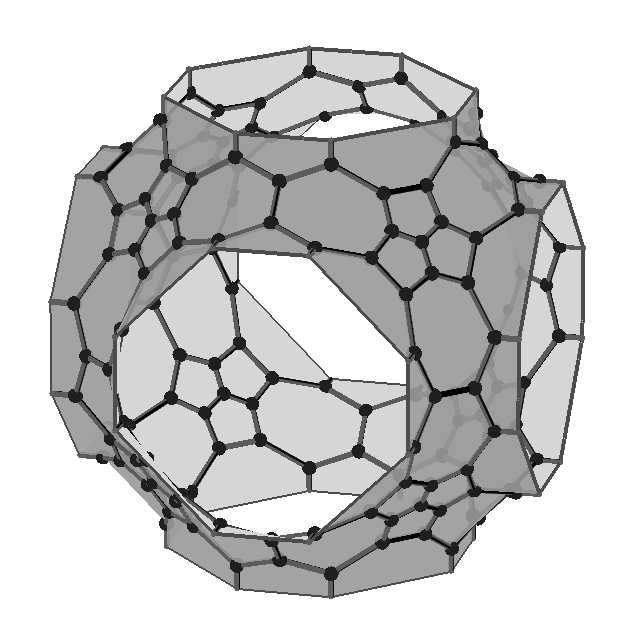}
    &
    \includegraphics[bb=0 0 640 640,scale=\scalebase]{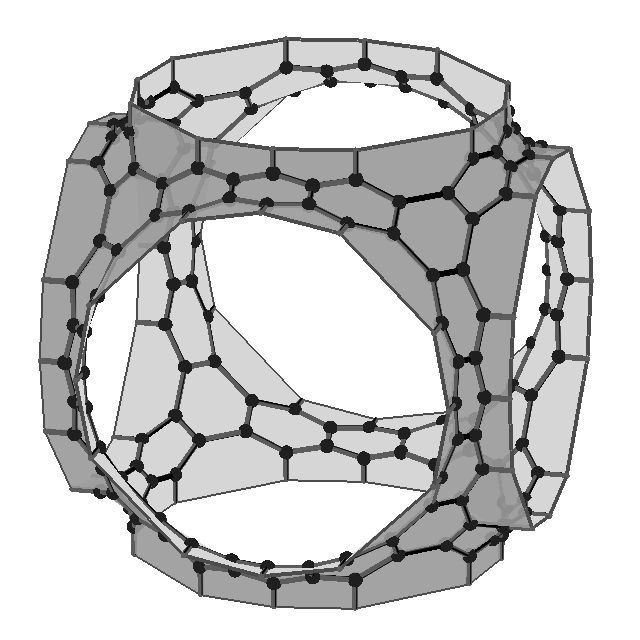}
    \cr
    &
    \includegraphics[bb=0 0 640 640,scale=\scalebase]{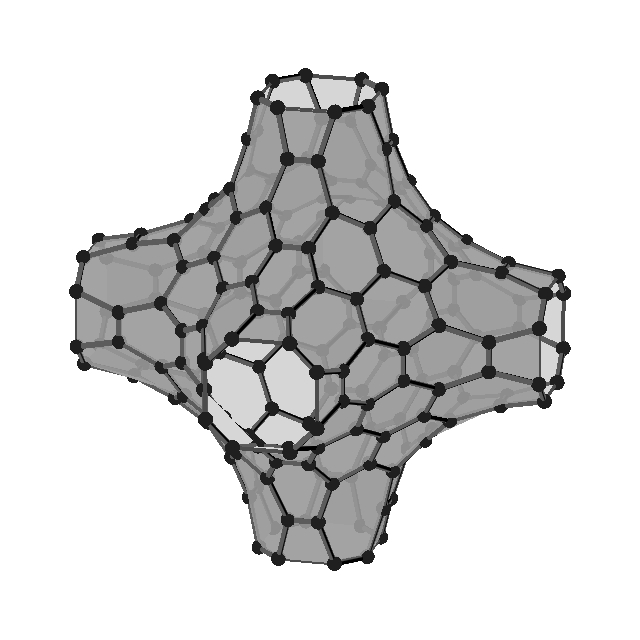}
    &
    \includegraphics[bb=0 0 640 640,scale=\scalebase]{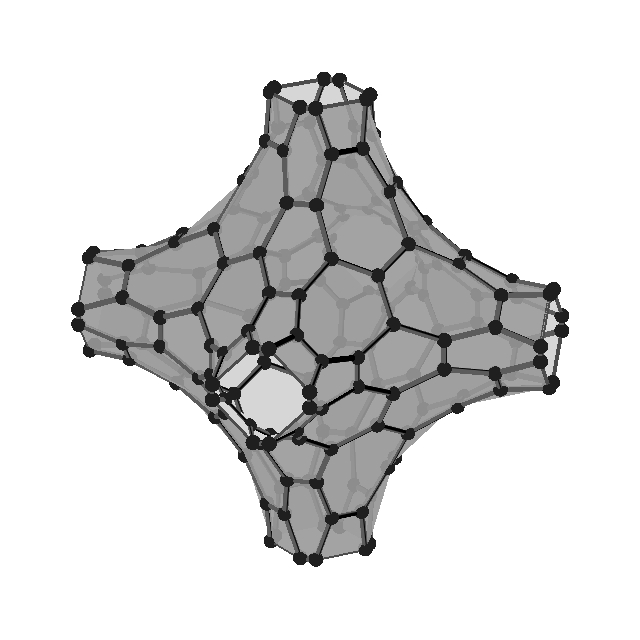}
    &
    \includegraphics[bb=0 0 640 640,scale=\scalebase]{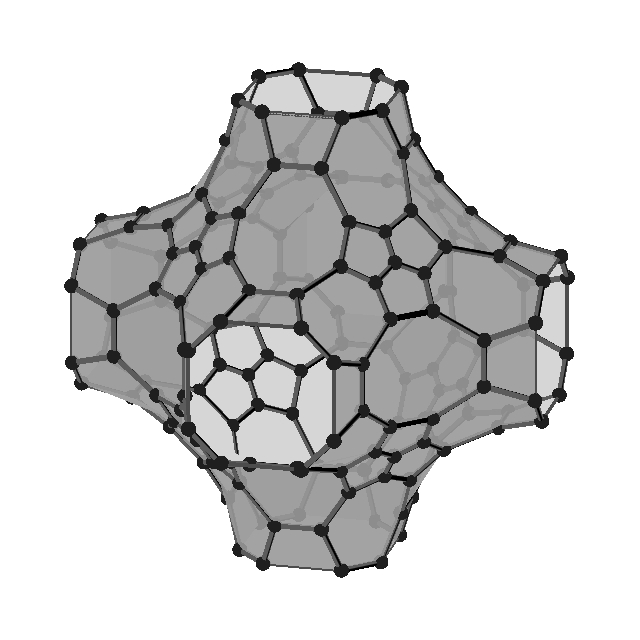}
    &
    \includegraphics[bb=0 0 640 640,scale=\scalebase]{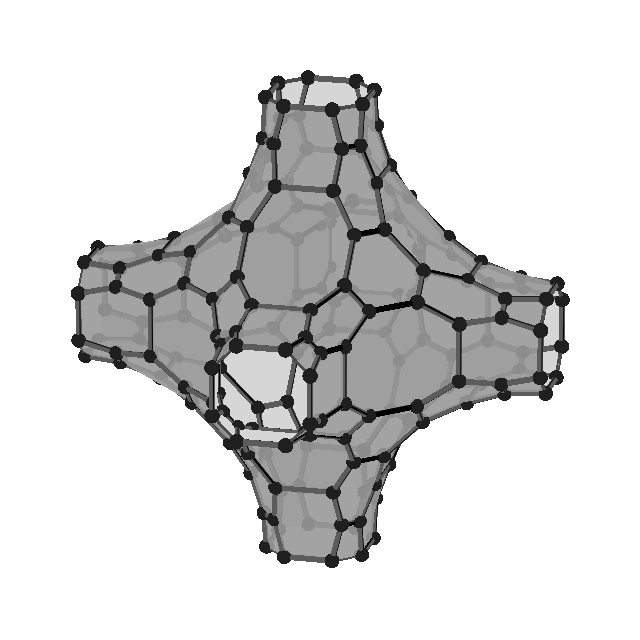}
    \cr
    \hline
    \vspace{0.01pt} & & & 
    \cr
    (d)
    &
    \includegraphics[bb=0 0 640 640,scale=\scalebase]{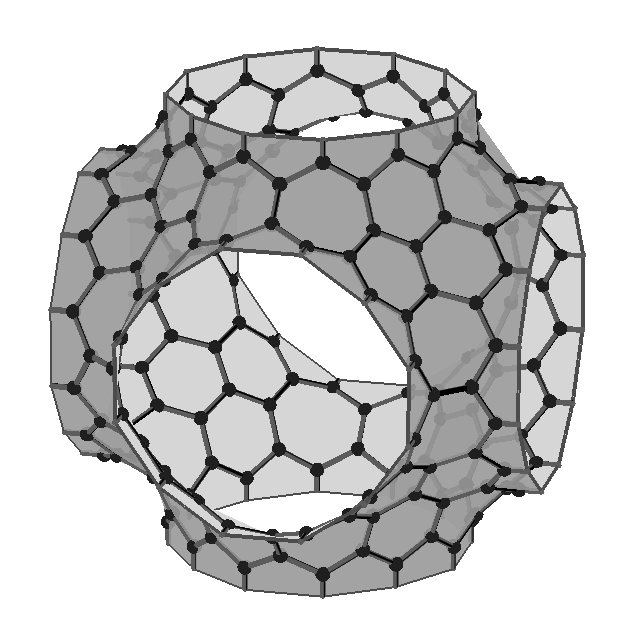}
    &
    \includegraphics[bb=0 0 640 640,scale=\scalebase]{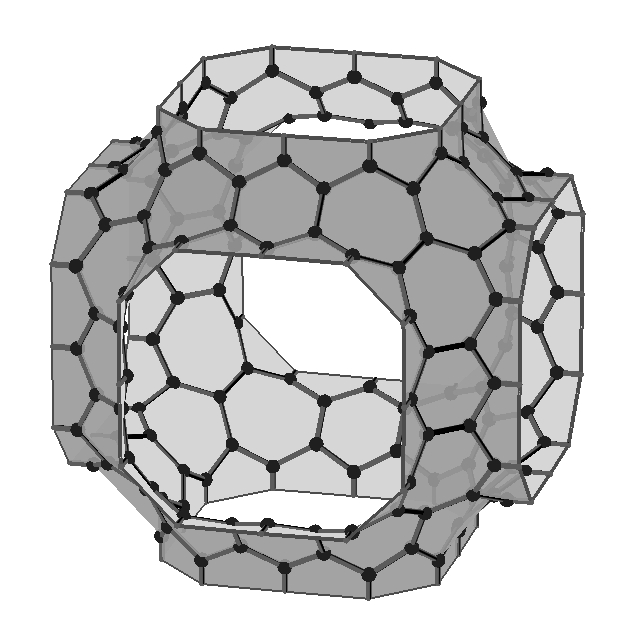}
    &
    \includegraphics[bb=0 0 640 640,scale=\scalebase]{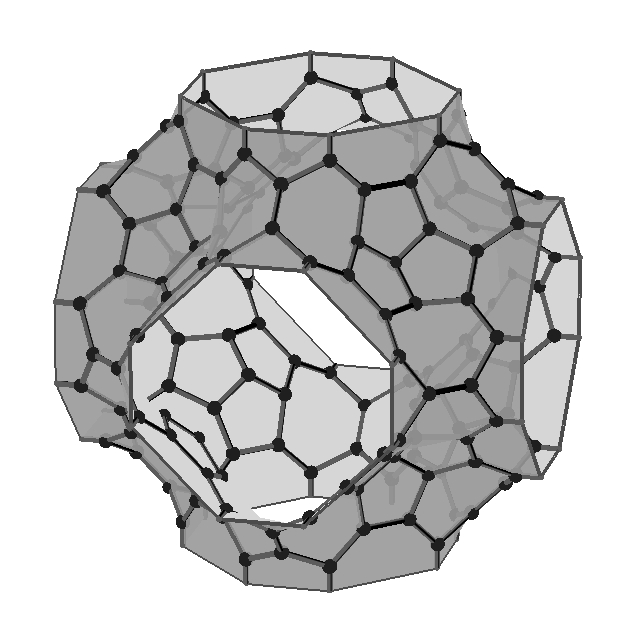}
    &
    \includegraphics[bb=0 0 640 640,scale=\scalebase]{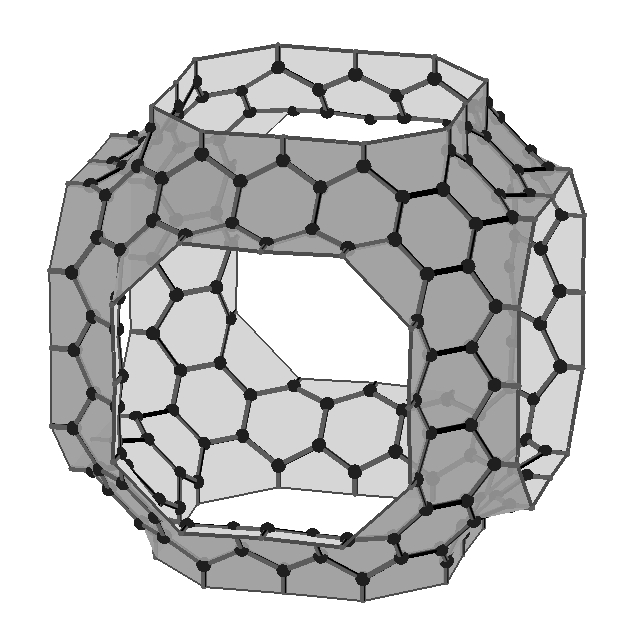}
    \cr
    &
    \includegraphics[bb=0 0 640 640,scale=\scalebase]{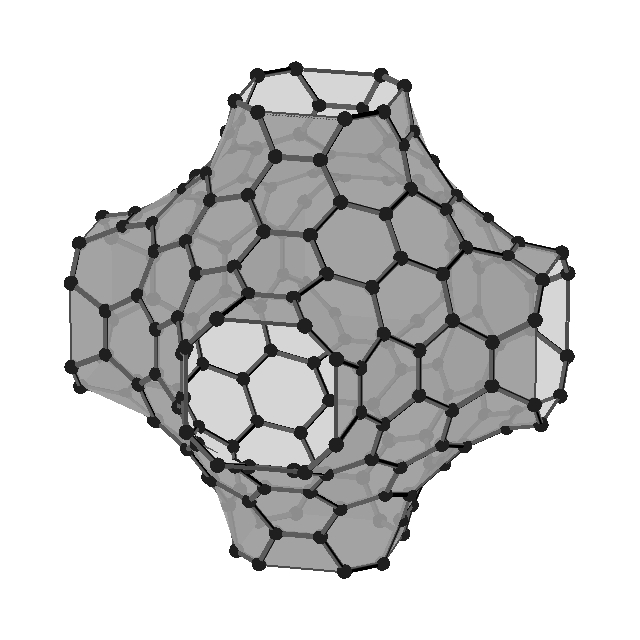}
    &
    \includegraphics[bb=0 0 640 640,scale=\scalebase]{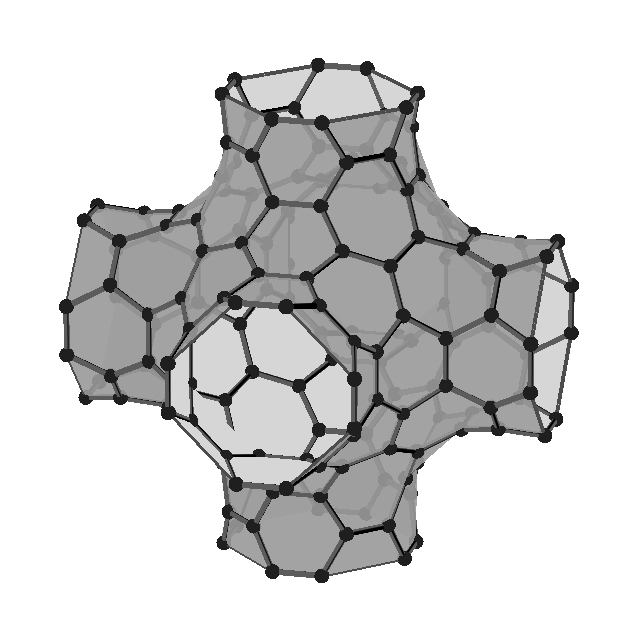}
    &
    \includegraphics[bb=0 0 640 640,scale=\scalebase]{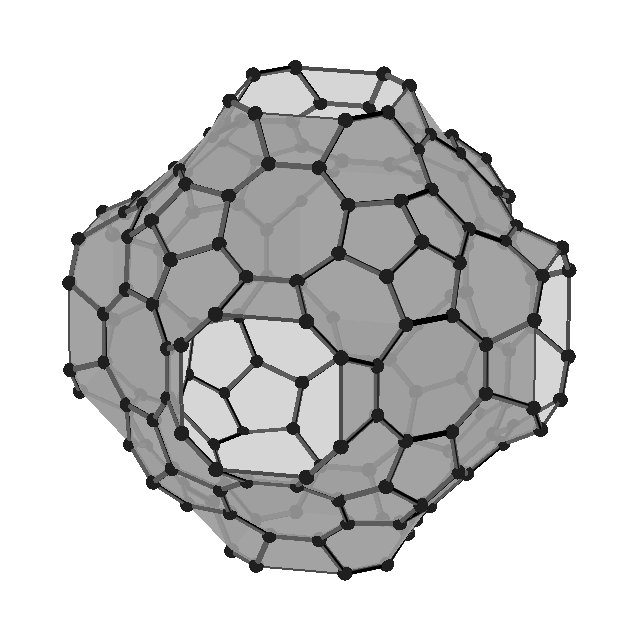}
    &
    \includegraphics[bb=0 0 640 640,scale=\scalebase]{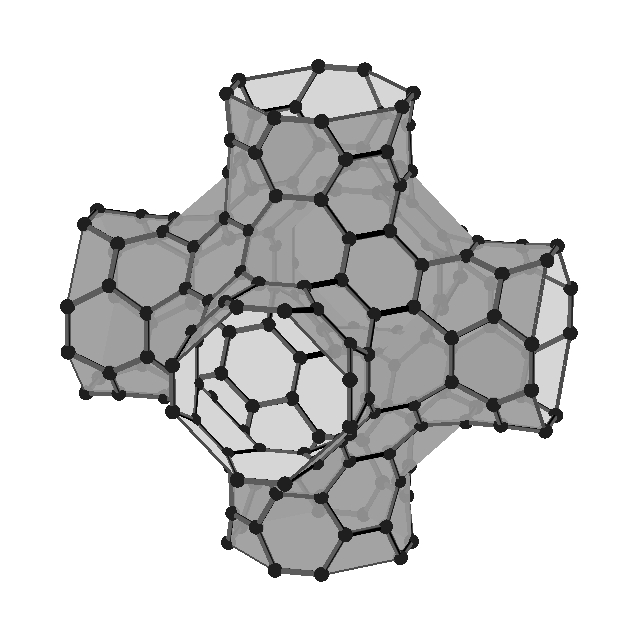}
    \cr
  \end{tabular}
  \caption{
    (a) The network in the kite region of each structure.
    (b) The network in the hexagonal regions.
    (c) The standard realization.
    (d) The relaxed configuration.
    Each figure in the row (c) or (d) is the same structure, but different figure.
    Shifting unit cell by $(1/2)(\vv{e}_x + \vv{e}_y + \vv{e}_z)$, 
    a figure in the row (c) or (d) changes to another one.
  }
  \label{fig:result}
\end{figure}

\section{Further problems}
Finally, we note further problems about negatively curved carbon structures.
\par
In addition to the Schwarz P-surface, 
there are well-known triply periodic minimal surfaces, 
namely, the Schwarz D-surface and G-surface (gyroid surface), 
and there are many studies about {\sptwo} carbon schwarzites related to D- and G-surfaces, 
for example \cite{Hyde-1, Hyde-2, Lenosky, H.Terrones, Other1}.
In this paper, 
we discuss trivalent networks, which reticulate the $P$-surface, 
and construct examples of physically stable carbon {\sptwo} structures.
Now, we are interested in the following problem:
Whether we can construct examples of physically stable carbon {\sptwo} structures 
which reticulate D- or G-surfaces using our method?
For example, 
calculate the physical stability of 
the carbon {\sptwo} structures that is constructed from 
\name{6-1-1}, \name{6-1-2}, \name{6-1-3} or \name{8-4-2} and
reticulates the D- or G-surface. 
\par
In this note, we claim that a discrete surface is negatively curved, 
if and only if the total discrete curvature is negative.
In other words, 
a discrete surface is negatively curved
if the structure is placed on a surface with negative curvature.
The first definition is not a point-wise characterization of negativity, 
and the last one is not rigorous.
On the other hand, 
there are definitions of curvature for general discrete surfaces.
For examples, 
\cite{Muller, Bobenko, Mayer}.
One well-known definition of curvature is 
\begin{equation}
  \label{eq:futher:1}
  K(\vv{x}) = 2 \pi - \sum_{j=0}^{k-1} \theta_j, 
  \quad
  \theta_j = \angle \vv{x}_j \vv{x} \vv{x}_{j+1}
\end{equation}
where $\{\vv{x}_j\}_{j=0}^{k-1}$ are the neighbours of $\vv{x}$ (see \cite{Mayer, Karcher-Polthier}).
However, since the structures
that we consider
are trivalent and they are standard realizations, 
each vertex and its neighbours are co-planer, 
therefore $K(\vv{x}) \equiv 0$ when calculated by (\ref{eq:futher:1}).
Definitions of curvature in \cite{Muller, Bobenko} assume meshed surfaces.
Hence, these definitions do not apply to
trivalent standard realizations of discrete surfaces.
We should consider how to define the curvature of trivalent standard realizations of discrete surfaces.
\section*{Appendix: Indexes of single wall nanotubes}
\label{sec:appendix}
We summarize the geometric structure of single wall nanotubes (SWNTs).
Mathematically, a SWNT is considered as the fundamental region of 
the $\Z$-action on the standard realization of the hexagonal lattice.
In the followings, we explain the geometric structures of SWNTs.
\par
First, we define
\begin{displaymath}
  \begin{aligned}
    &\vv{v}_0 = (0, 0), 
    \quad
    \vv{v}_1 = (-\sqrt{3}/2, 1/2), 
    \quad
    \vv{v}_2 = (\sqrt{3}, 1/2), 
    \quad
    \vv{v}_3 = (0, -1), 
    \\
    &\vv{a}_1 
    = \vv{v}_2 - \vv{v}_1 = (\sqrt{3}, 0)
    \quad
    \vv{a}_2 
    = \vv{v}_3 - \vv{v}_1 = (\sqrt{3}/2, -3/2), 
  \end{aligned}
\end{displaymath}
then the graph $X_0 = (V_0, E_0)$, 
$V_0 = \{\vv{v}_i\}_{i=0}^3$, 
$E_0 = \{(\vv{v}_0, \vv{v}_i)\}_{i=1}^3$
is the fundamental region of the hexagonal lattice, 
and $\{\vv{a}_1, \vv{a}_2\}$ is the basis of the parallel transformations
(see Figure \ref{fig:SWNT}).
Here, the angle between $\vv{a}_1$ and $\vv{a}_2$ is $\pi/3$.
By using the basis, 
we may write all vertices $\vv{v}$ of the hexagonal lattice as
\begin{displaymath}
  \vv{v} = \alpha_1 \vv{a}_1 + \alpha_2 \vv{a}_2, 
  \quad
  (\alpha_1, \alpha_2) \in \Z \times \Z, 
  \quad
  \text{or}
  \quad
  (\alpha_1, \alpha_2) \in (\Z-1/3) \times (\Z-1/3).
\end{displaymath}
\begin{definition*}
  \label{definition:chiral}
  The vector $\vv{c} = c_1 \vv{a}_1 + c_2 \vv{a}_2$ is called 
  the \emph{chiral vector} 
  and $(c_1, c_2)$ is called the \emph{chiral index}
  if and only if
  the SWNT is constructed from the hexagonal lattice by identifying 
  $\vv{x}$ and $\vv{x} + \vv{c}$.
  \par
  The vector 
  $\vv{t} = t_1 \vv{a}_1 + t_2 \vv{a}_2$, 
  where $(t_1, t_2) = ((c_1 + 2 c_2)/d(c), -(2c_1 + c_2)/d(\vv{c}))$, 
  and 
  $d(\vv{c}) = \gcd(c_1 + 2 c_2, 2c_1 + c_2)$,  
  is called the \emph{lattice vector} of the SWNT with chiral index $\vv{c} = (c_1, c_2)$ (see Figure \ref{fig:SWNT}).
\end{definition*}
The chiral vector indicates the direction of the circle of the SWNT, 
and is the period vector of the $\Z$-action on the hexagonal lattice.
The chiral vector is orthogonal to the lattice vector, 
and the lattice vector is the minimum period of the hexagonal lattice
along the tube axis of the SWNT.
The circumference $L(\vv{c})$ of the SWNT is derived from the chiral index by
\begin{math}
  L(\vv{c}) 
  = |\vv{c}|
  = \sqrt{3(c_1^2 + c_2^2 + c_1 c_2)}.
\end{math}
The electronic properties of single wall carbon nanotubes depends on
the chiral index (see \cite{Saito}).
\par
Next, let us consider finite length single wall nanotubes.
More precisely, consider a finite length SWNT, 
which terminates at the vertices (atoms) 
with $\vv{x}_0 = \vv{0}$ and 
$\vv{x}_1 = \alpha_1 \vv{a}_1 + \alpha_2 \vv{a}_2$, 
and characterize the length of this SWNT by $\alpha_1$ and $\alpha_2$.
In the followings, 
$\SWNT(\vv{c}, \alpha_1 \vv{a}_1 + \alpha_2 \vv{a}_2)$ 
denotes the SWNT which is terminated by $\vv{0}$ and $\alpha_1 \vv{a}_1 + \alpha_2 \vv{a}_2$ and whose chiral vector is $\vv{c}$.
It is easy to show that 
the length of $\SWNT(\vv{c}, \vv{x})$ is $\left|\inner{\vv{x}}{\vv{e}_t}\right|$, 
where $\vv{e}_t = \vv{t}/|\vv{t}|$.
In \cite{Isobe-Naito}, 
we define the \emph{length index} $\ell(\vv{c}, \vv{x})$ 
of $\SWNT(\vv{c}, \vv{x})$ as
\begin{displaymath}
  \ell(\vv{c}, \vv{x})
  = 
  \left|\inner{\vv{x}}{\vv{e}_t}\right|/\sqrt{3}
  =
  \frac{\sqrt{3}|c_2 \alpha_1 - c_1 \alpha_2|}{2 \sqrt{c_1^2 + c_2^2 + c_1 c_2}}, 
  \quad
  \vv{x} = \alpha_1 \vv{a}_1 + \alpha_2 \vv{a}_2.
\end{displaymath}
The length index corresponds to how many hexagons are arranged in the direction of the tube axis.
\par
Now, we consider 
$\SWNT(\vv{c}, \vv{t})$, 
which has 
the canonical length for the given chiral index $\vv{c}$.
By using the definition of $\vv{t}$, 
we obtain the length index $\ell(\vv{c}, \vv{t})$ of $\SWNT(\vv{c}, \vv{t})$ as
\begin{displaymath}
  \ell(\vv{c}, \vv{x})
  =
  \frac{\sqrt{3(c_1^2 + c_2^2 + c_1 c_2)}}{d(\vv{c})}
  =
  \frac{L(\vv{c})}{d(\vv{c})}, 
\end{displaymath}
and its area as
\begin{displaymath}
  \Area(\SWNT(\vv{c}, \vv{t}))
  =
  \sqrt{3} \ell(\vv{c}, \vv{t}) L(\vv{c})
  =
  \frac{\sqrt{3} L(\vv{c})}{d(\vv{c})}.
\end{displaymath}
\par
Finally, we calculate the number of hexagons in the fundamental region of 
$\SWNT(\vv{c}, \vv{t})$, 
which is the fundamental region of action generated by the lattice 
$\{\vv{c}, \vv{t}\}$.
Since all hexagons in the lattice are congruent and their volume are 
$3\sqrt{3}/2$, 
the fundamental region contains $F$ hexagons, where
\begin{equation}
  \label{eq:SWNT:F}
 F = \frac{2L(\vv{c})^2}{3d(\vv{c})}.
\end{equation}
Combining (\ref{eq:VEF}) and (\ref{eq:SWNT:F}), 
we obtain the number of vertices $V$, of edges $E$, and of faces $F$ 
in Table \ref{table:euler}.
%
\begin{figure}
  \centering
  \includegraphics[scale=0.75,clip=true]{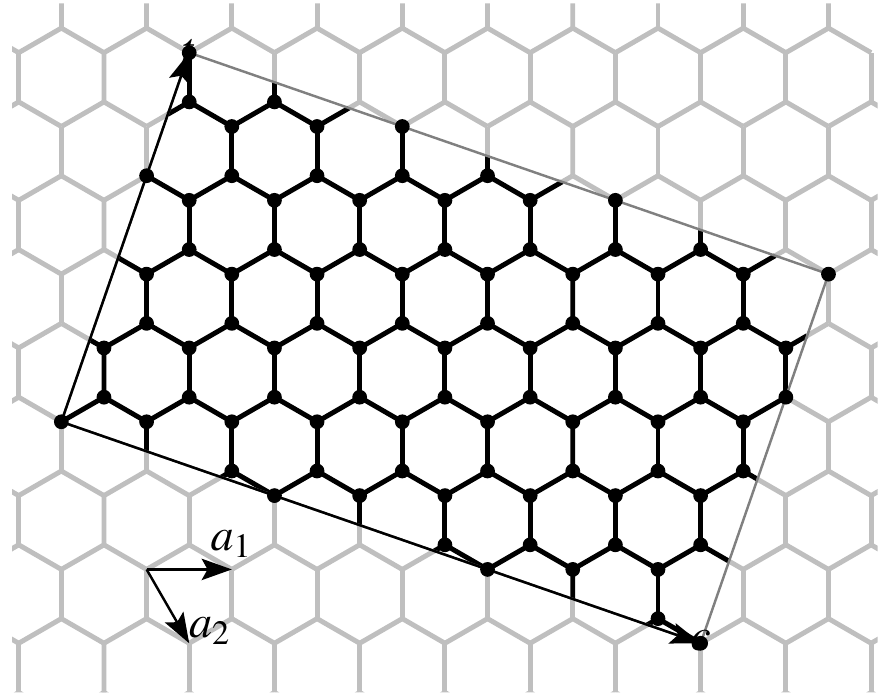}
  \caption{
    Configuration of the single-wall carbon nanotube, 
    with $\vv{c} = (6,3)$ and $\vv{t} = (4, -5)$.
  }
  \label{fig:SWNT}
\end{figure}

\begin{mynote}{Note about figures}
  Part of Figure \ref{fig:result} and Figure \ref{fig:schwartz-P}(a) are 
  gray scale version of figures published in \cite{Naito-Carbon}.
  Figures \ref{fig:schwartz-P}(a), (b), 
  Figure \ref{fig:8-2-1}
  and Figures \ref{fig:result}(c), (d) of \name{6-1-1-P} 
  are also published in \cite{Naito-Suurikagaku}.
\end{mynote}

\begin{acknowledgement}
  The author gratefully thanks to 
  Prof. Davide M. Proserpio who informed us 
  about articles on this subject and that
  \name{6-1-2-P} is almost the same structure as C152 in \cite{Park}.
  The author also thanks 
  Professor Motoko Kotani and Professor Yasumasa Nishiura.
  They gave the author an opportunity to talk at the Symposium 
  ``Mathematical Challenge to a New Phase of Materials Science''
  in Kyoto, 2014.
  The author was partially supported by Grants-in-Aid for Scientific Research (C) 40211411 and (A) 15H02055.
\end{acknowledgement}

\begin{mynote}{Dedication}
  I dedicate this paper to my late wife Yumiko Naito (September 8 1963 -- October 2 2015). 
  While battling breast cancer, she spent her days with a positive and enthusiastic attitude, supporting my research and caring for our son.
  My research up to now could not have existed without her support.
  I dedicate this paper to Yumiko's support during 24 years of marriage and to the wonderful memories of the 30 years since we first me.
  To my wife Yumiko, with gratitude and love.
\end{mynote}


\begin{thebibliography}{10}

\bibitem{Naito-Carbon}
M.~Tagami, Y.~Liang, H.~Naito, Y.~Kawazoe, and M.~Kotani.
\newblock Negatively curved cubic carbon crystals with octahedral symmetry.
\newblock {\em Carbon}, {\bfseries 76}, 266--274, (2014).
\newblock \hfill\break \href{http://dx.doi.org/10.1016/j.carbon.2014.04.077}{http://dx.doi.org/10.1016/j.carbon.2014.04.077}.

\bibitem{Mackay}
A.~L. Mackay and H.~Terrones.
\newblock Diamond from graphite.
\newblock {\em Nature}, {\bfseries 352}, 762, (1991).
\newblock \hfill\break \href{http://dx.doi.org/10.1038/352762a0}{http://dx.doi.org/10.1038/352762a0}.

\bibitem{H.Terrones}
H.~Terrones and A.~L. Mackay.
\newblock Triply periodic minimal-surfaces decorated with curved graphite.
\newblock {\em Chem. Phys. Lett.}, {\bfseries 207}, 45--50, (1993).
\newblock \hfill\break \href{http://dx.doi.org/10.1016/0009-2614(93)85009-d}{http://dx.doi.org/10.1016/0009-2614(93)85009-d}.

\bibitem{H.Terrones3}
H.~Terrones and A.~L. Mackay.
\newblock Negatively curved graphite and triply periodic minimal-surfaces.
\newblock {\em J. Math. Chem.}, {\bfseries 15}, 183--195, (1994).
\newblock \href{http://dx.doi.org/10.1007/bf01277558}{http://dx.doi.org/10.1007/bf01277558}.

\bibitem{Lenosky}
T.~Lenosky, X.~Gonze, M.~Teter, and V.~Elser.
\newblock Energetics of negatively curved graphitic carbon.
\newblock {\em Nature}, {\bfseries 355}, 333--335, (1992).
\newblock \href{http://dx.doi.org/10.1038/355333a0}{http://dx.doi.org/10.1038/355333a0}.

\bibitem{R.Phillips}
R.~Phillips, D.~A. Drabold, T.~Lenosky, G.~B. Adams, and O.~F. Sankey.
\newblock Electronic-structure of schwarzite.
\newblock {\em Phys. Rev. B}, {\bfseries 46}, 1941--1943, (1992).
\newblock \hfill\break \href{http://dx.doi.org/10.1103/PhysRevB.46.1941}{http://dx.doi.org/10.1103/PhysRevB.46.1941}.

\bibitem{M-Z.Huang}
M.~Z. Huang, W.~Y. Ching, and T.~Lenosky.
\newblock Electronic-properties of negative-curvature periodic graphitic carbon surfaces.
\newblock {\em Phys. Rev. B}, {\bfseries 47}, 1593--1606, (1993).
\newblock \hfill\break \href{http://dx.doi.org/10.1103/PhysRevB.47.1593}{http://dx.doi.org/10.1103/PhysRevB.47.1593}.

\bibitem{Townsend}
S.~J. Townsend, T.~J. Lenosky, D.~A. Muller, C.~S. Nichols, and V.~Elser.
\newblock Negatively curved graphitic sheet model of amorphous-carbon.
\newblock {\em Phys. Rev. Lett.}, {\bfseries {69}}, {921--924}, (1992).
\newblock \hfill\break \href{http://dx.doi.org/10.1103/PhysRevLett.69.921}{http://dx.doi.org/10.1103/PhysRevLett.69.921}.

\bibitem{Mackay-Fowler}
A.~L. Mackay, H.~Terrones, and P.~W. Fowler.
\newblock Hypothetical graphite structures with negative gaussian curvature.
\newblock {\em Phil. Trans. R. Soc. A}, {\bfseries 343}, {113--127}, (1993).
\newblock \hfill\break \href{http://dx.doi.org/10.1098/rsta.1993.0045}{http://dx.doi.org/10.1098/rsta.1993.0045}.

\bibitem{Kotani-Sunada}
M.~Kotani and T.~Sunada.
\newblock Standard realizations of crystal lattices via harmonic maps.
\newblock {\em Trans. Amer. Math. Soc.}, {\bfseries 353}, 1--20, (2001).
\newblock \href{http://dx.doi.org/10.1090/S0002-9947-00-02632-5}{http://dx.doi.org/10.1090/S0002-9947-00-02632-5}.

\bibitem{Sunada-Book}
T.~Sunada.
\newblock {\em Topological crystallography}, volume~6 of {\em Surveys and Tutorials in the Applied Mathematical Sciences}.
\newblock Springer, Tokyo, (2013).
\newblock \hfill\break \href{http://dx.doi.org/10.1007/978-4-431-54177-6}{http://dx.doi.org/10.1007/978-4-431-54177-6}.

\bibitem{Schoen}
A.~H. Schoen.
\newblock Infinite periodic minimal surfaces without self-intersections, 
\newblock NASA Technical Note, TN D-5541, (1970).

\bibitem{Schoen2}
A.~H. Schoen.
\newblock Reflections concerning triply-periodic minimal surfaces.
\newblock {\em Interface Focus}, {\bfseries {2}}, {658--668}, (2012).
\newblock \href{http://dx.doi.org/10.1098/rsfs.2012.0023}{http://dx.doi.org/10.1098/rsfs.2012.0023}.

\bibitem{Karcher-Polthier}
H.~Karcher and K.~Polthier.
\newblock Construction of triply periodic minimal surfaces.
\newblock {\em Philos. Trans. Roy. Soc. London Ser. A}, {\bfseries 354},
  2077--2104, (1996).
\newblock \href{http://dx.doi.org/10.1098/rsta.1996.0093}{http://dx.doi.org/10.1098/rsta.1996.0093}.

\bibitem{Molnar}
E.~Molnar.
\newblock On triply periodic minimal balance surfaces.
\newblock {\em Struct. Chem.}, {\bfseries {13}}, {267--275}, (2002).
\newblock \hfill\break \href{http://dx.doi.org/10.1023/A:1015855721911}{http://dx.doi.org/10.1023/A:1015855721911}.

\bibitem{Other2}
V.~Rosato, M.~Celino, S.~Gaito, and G.~Benedek.
\newblock Thermodynamic behavior of a carbon schwarzite.
\newblock {\em Comput. Mater. Sci.}, {\bfseries {20}}, {387--393}, (2001).
\newblock \hfill\break \href{http://dx.doi.org/10.1016/S0927-0256(00)00197-X}{http://dx.doi.org/10.1016/S0927-0256(00)00197-X}.

\bibitem{Other3}
M.~Homyonfer, Y.~Feldman, L.~Margulis, and R.~Tenne.
\newblock Negative curvature in inorganic fullerene-like structure.
\newblock {\em Fullerene Sci. Techn.}, {\bfseries {6}}, {59--66}, (1998).
\newblock \hfill\break \href{http://dx.doi.org/10.1080/10641229809350185}{http://dx.doi.org/10.1080/10641229809350185}.

\bibitem{Other4}
D.~Vanderbilt and J.~Tersoff.
\newblock Negative-curvature fullerene analog of {C}60.
\newblock {\em Phys. Rev. Lett.}, {\bfseries {68}}, {511--513}, (1992).
\newblock \href{http://dx.doi.org/10.1103/PhysRevLett.68.511}{http://dx.doi.org/10.1103/PhysRevLett.68.511}.

\bibitem{Other5}
A.~Ceulemans, R.~B. King, S.~A. Bovin, K.~M. Rogers, A.~Troisi, and P.~W.~Fowler.
\newblock The heptakisoctahedral group and its relevance to carbon allotropes with negative curvature.
\newblock {\em J. Math. Chem.}, {\bfseries {26}}, {101--123}, (1999).
\newblock \href{http://dx.doi.org/10.1023/A:1019129827020}{http://dx.doi.org/10.1023/A:1019129827020}.

\bibitem{Other6}
R.~B. King.
\newblock Chemical applications of topology and group theory .29. low density polymeric carbon allotropes based on negative curvature structures.
\newblock {\em J. Phys. Chem}, {\bfseries {100}}, {15096--15104}, (1996).
\newblock \href{http://dx.doi.org/10.1021/jp9613201}{http://dx.doi.org/10.1021/jp9613201}.

\bibitem{Delgado-Friedrichs-2}
O.~Delgado-Friedrichs.
\newblock {Equilibrium placement of periodic graphs and convexity of plane tilings}.
\newblock {\em Discrete Comput. Geom.}, {\bfseries {33}}, {67--81}, (2005).
\newblock \hfill\break \href{http://dx.doi.org/10.1007/s00454-004-1147-x}{http://dx.doi.org/10.1007/s00454-004-1147-x}.

\bibitem{Delgado-Friedrichs-1}
O.~Delgado-Friedrichs and M.~O'Keeffe.
\newblock {Identification of and symmetry computation for crystal nets}.
\newblock {\em Acta Crystallogr. A}, {\bfseries {59}}, {351--360}, (2003).
\newblock \hfill\break \href{http://dx.doi.org/10.1107/S0108767303012017}{http://dx.doi.org/10.1107/S0108767303012017}.

\bibitem{Naito-AMS}
H.~Naito.
\newblock Visualization of standard realized crystal lattices.
\newblock In {\em Spectral analysis in geometry and number theory}, volume 484
  of {\em Contemp. Math.}, pages 153--164. Amer. Math. Soc., Providence, RI,
  (2009).
\newblock \href{http://dx.doi.org/10.1090/conm/484/09472}{http://dx.doi.org/10.1090/conm/484/09472}.

\bibitem{Sunada-Notice}
T.~Sunada.
\newblock Crystals that nature might miss creating.
\newblock {\em Notices Amer. Math. Soc.}, {\bfseries 55}, 208--215, (2008).
\newblock Correction:''Crystals that nature might miss creating'', {\em ibid.}
  {\bfseries 55}, 343, (2008).

\bibitem{Hyde-2}
S.~T. Hyde and S.~Ramsden.
\newblock {Polycontinuous morphologies and interwoven helical networks}.
\newblock {\em Europhys. Lett.}, {\bfseries {50}}, {135--141}, (2000).
\newblock \href{http://dx.doi.org/10.1209/epl/i2000-00245-y}{http://dx.doi.org/10.1209/epl/i2000-00245-y}.

\bibitem{Hyde-3}
S.~T. Hyde, M.~O'Keeffe, and D.~M. Proserpio.
\newblock {A short history of an elusive yet ubiquitous structure in chemistry, materials, and mathematics}.
\newblock {\em Angew. Chem. Int. Edit.}, {\bfseries {47}}, {7996--8000},
  (2008).
\newblock \href{http://dx.doi.org/10.1002/anie.200801519}{http://dx.doi.org/10.1002/anie.200801519}.

\bibitem{K4}
M.~Itoh, M.~Kotani, H.~Naito, T.~Sunada, Y.~Kawazoe, and T.~Adschiri.
\newblock New metallic carbon crystal.
\newblock {\em Phys. Rev. Lett.}, {\bfseries 102}, 055703, (2009).
\newblock \hfill\break \href{http://dx.doi.org/10.1103/PhysRevLett.102.055703}{http://dx.doi.org/10.1103/PhysRevLett.102.055703}.

\bibitem{D-Friedrichs}
O.~Delgado-Friedrichs, M.~O. O'Keeffe, and O.~M. Yaghi.
\newblock Three-periodic nets and tilings: semiregular nets.
\newblock {\em Acta Crystallogr. A}, {\bfseries 59}, 515--525, (2003).
\newblock \hfill\break \href{http://dx.doi.org/10.1107/S0108767303017100}{http://dx.doi.org/10.1107/S0108767303017100}.

\bibitem{Conway-Huson}
J.~H. Conway and D.~H. Huson.
\newblock {The orbifold notation for two-dimensional groups}.
\newblock {\em Struct Chem.}, {\bfseries {13}}, {247--257}, (2002).
\newblock \href{http://dx.doi.org/10.1023/A:1015851621002}{http://dx.doi.org/10.1023/A:1015851621002}.

\bibitem{epinet}
{A}ustralian~{N}ational {U}niversity.
\newblock {EPINET} {P}roject.
\newblock \href{http://epinet.anu.edu.au/}{http://epinet.anu.edu.au/}

\bibitem{Hyde-4}
S.~J. Ramsden, V.~Robins, and S.~T. Hyde.
\newblock {Three-dimensional Euclidean nets from two-dimensional hyperbolic
  tilings: kaleidoscopic examples}.
\newblock {\em Acta Crystallogr. A}, {\bfseries {65}}, {81--108}, (2009).
\newblock \href{http://dx.doi.org/10.1107/S0108767308040592}{http://dx.doi.org/10.1107/S0108767308040592}.

\bibitem{Park}
S.~Park, K.~Kittimanapun, J.~S. Ahn, Y.-K. Kwon, and D.~Tom{\'a}nek.
\newblock Designing rigid carbon foams.
\newblock {\em J. Phys.:Condens. Matter}, {\bfseries {22}}, (2010).
\newblock \hfill\break \href{http://dx.doi.org/10.1088/0953-8984/22/33/334220}{http://dx.doi.org/10.1088/0953-8984/22/33/334220}.

\bibitem{Hyde-1}
G.~E. Schroder, S.~J. Ramsden, A.~G. Christy, and S.~T. Hyde.
\newblock {Medial surfaces of hyperbolic structures}.
\newblock {\em Eur. Phys. J. B}, {\bfseries {35}}, {551--564}, (2003).
\newblock \hfill\break \href{http://dx.doi.org/10.1140/epjb/e2003-00308-y}{http://dx.doi.org/10.1140/epjb/e2003-00308-y}.

\bibitem{Other1}
I.~Spagnolatti, M.~Bernasconi, and G.~Benedek.
\newblock Electron-phonon interaction in carbon schwarzites.
\newblock {\em Europ. Phys. J. B}, {\bfseries {32}}, {181--187}, (2003).
\newblock \hfill\break \href{http://dx.doi.org/10.1140/epjb/e2003-00087-5}{http://dx.doi.org/10.1140/epjb/e2003-00087-5}.

\bibitem{Muller}
C.~M{\"u}ller and J.~Wallner.
\newblock Oriented mixed area and discrete minimal surfaces.
\newblock {\em Discrete Comput. Geom.}, {\bfseries 43}, 303--320, (2010).
\newblock \href{http://dx.doi.org/10.1007/s00454-009-9198-7}{http://dx.doi.org/10.1007/s00454-009-9198-7}.

\bibitem{Bobenko}
A.~I. Bobenko, H.~Pottmann, and J.~Wallner.
\newblock A curvature theory for discrete surfaces based on mesh parallelity.
\newblock {\em Math. Ann.}, {\bfseries 348}, 1--24, (2010).
\newblock \hfill\break \href{http://dx.doi.org/10.1007/s00208-009-0467-9}{http://dx.doi.org/10.1007/s00208-009-0467-9}.

\bibitem{Mayer}
M.~Meyer, M.~Desbrun, P.~Schr{\"o}der, and A.~H. Barr.
\newblock Discrete differential-geometry operators for triangulated 2-manifolds, (2002).
\newblock Visualization and Mathematica III.

\bibitem{Saito}
R.~Saito, M.~Fujita, G.~Dresselhaus, and M.~S. Dresselhaus.
\newblock Electronic structure of chiral graphene tubules.
\newblock {\em Appl. Phys. Lett.}, {\bfseries {60}}, {2204--2206}, (1992).
\newblock \hfill\break \href{http://dx.doi.org/10.1063/1.107080}{http://dx.doi.org/10.1063/1.107080}.

\bibitem{Isobe-Naito}
T.~Matsuno, H.~Naito, S.~Hitosugi, S.~Sato, M.~Kotani, and H.~Isobe.
\newblock Geometric measures of finite carbon nanotube molecules: a proposal
  for length index and filling indexes.
\newblock {\em Pure Appl. Chem.}, {\bfseries {86}}, {489--495}, (2014).
\newblock \href{http://dx.doi.org/10.1515/pac-2014-5006}{http://dx.doi.org/10.1515/pac-2014-5006}.

\bibitem{Naito-Suurikagaku}
H.~Naito.
\newblock {C}hemistry and mathematics -- discrete geometry and carbon structures.
\newblock {\em Mathematical Sciences (SUURI KAGAKU)}, {\bfseries 624}, 42--47,
  (2015).
\newblock SAIENSU-SHA CO., LTD., (in Japanese).

\end{thebibliography}


\end{document}